\documentclass[a4paper,10pt]{article}

\usepackage{lineno,hyperref}
\modulolinenumbers[5]
\usepackage{amssymb}
\usepackage[margin=1in]{geometry}
\usepackage{amsthm}
\usepackage{graphicx}
\usepackage{chemarr}
\usepackage{bbm}
\usepackage{color}
\usepackage[colorinlistoftodos]{todonotes}
\usepackage{mathtools}

\newcommand{\cut}[1]{}

\newcommand{\react}[2]{#1\rightarrow#2}

\newcommand{\Rxn}{\mathcal{R}}
\newcommand{\Int}{\mathbb{Z}}
\newcommand{\Net}{\mathcal{N}}
\newcommand{\Spc}{\mathcal{S}}
\newcommand{\Vis}{\mathcal{V}}
\newcommand{\Hsp}{\mathcal{H}}

\newcommand{\Rel}{\mathbb{R}}
\newcommand{\Ord}{\mathcal{O}}

\newcommand{\Cmp}{\mathcal{C}}

\newcommand{\Fam}{\mathcal{F}}
\newcommand{\1}{\mathbbm{1}}
\newcommand{\Uni}{\mathcal{U}}
\newcommand{\tofrom}{\rightleftharpoons}
\newcommand{\DB}{\texttt{Full}}
\newcommand{\PMD}{\texttt{PointMass}}
\newcommand{\BiDB}{\texttt{Bimol}}
\newcommand{\StiDB}{\texttt{SpanTree}}
\newcommand{\PMCN}{\texttt{PointMassMix}}
\newcommand{\Mix}{\texttt{Mix}}
\newcommand{\UD}{\texttt{MultiDimUnif}}
\newcommand{\UDbis}{\UD'}
\newcommand{\PPN}{\texttt{ProdPois}}

\theoremstyle{plain}
\newtheorem{thm}{Theorem}
\newtheorem{lem}[thm]{Lemma}

\newtheorem{prop}[thm]{Proposition}
\newtheorem{athm}{Theorem}[section]
\newtheorem{alem}[athm]{Lemma}

\newtheorem{aprop}[athm]{Proposition}

\theoremstyle{definition}
\newtheorem{defn}{Definition}

\newtheorem{cons}{Construction}
\newtheorem{adefn}{Definition}[section]

\newenvironment{consbis}[1]
  {\renewcommand{\thecons}{\ref{#1}$'$}%
   \addtocounter{cons}{-1}%
   \begin{cons}}
  {\end{cons}}

\theoremstyle{remark}
\newtheorem{rmk}{Remark}
\newtheorem{armk}{Remark}[section]

\usepackage[numbers, sort]{natbib}

\begin{document}

\title{Stochastic Chemical Reaction Networks for Robustly Approximating Arbitrary Probability Distributions}

%
%
%

\author{Daniele Cappelletti\footnotemark[1]\phantom{,,}\footnotemark[2] \and Andr\'{e}s Ortiz-Mu\~{n}oz\footnotemark[1]\phantom{,,}\footnotemark[3] \and David F.\ Anderson \footnotemark[2]\and Erik Winfree \footnotemark[3]}

\footnotetext[1]{These authors contributed equally.}
\footnotetext[2]{University of Wisconsin-Madison, Madison, WI, USA.}
\footnotetext[3]{California Institute of Technology, Pasadena, CA, USA}

\maketitle

\begin{abstract}
We show that discrete distributions on the $d$-dimensional non-negative integer lattice can be approximated arbitrarily well via the marginals of stationary distributions for various classes of stochastic chemical reaction networks. We begin by providing a class of detailed balanced networks and prove that they can approximate any discrete distribution to any desired accuracy. However, these detailed balanced constructions rely on the ability to initialize a system precisely, and are therefore susceptible to perturbations in the initial conditions.  We therefore  provide another construction based on the ability to approximate point mass distributions and prove that this construction is capable of approximating arbitrary discrete distributions for any choice of initial condition.  In particular, the developed models are ergodic, so their limit distributions are robust to a finite number of perturbations over time in the counts of molecules.

\end{abstract}


\section{Introduction}

Chemical reaction networks (CRNs) with  mass-action kinetics \cite{anderson2015stochastic,gunawardena2003chemical} model the behavior of well-mixed chemical solutions and have a wide range of applications in science and engineering.  In particular, they are used to study the behavior of natural, industrial, and biological processes. Thus, it is important to understand their mathematical foundations. For systems where the number of molecules is large, stochasticity averages out so the state variables are the real-valued concentrations of molecular species, and dynamics are expressed as ordinary differential equations. These models are well-studied and much is known about their stationary behavior \cite{feinberg1987chemical}. Furthermore, the dynamics of deterministic CRNs are capable of simulating arbitrary electrical and digital circuits, and there is a sense in which they are capable of universal computation \cite{samardzija1989nonlinear,magnasco1997chemical}. Of increasing importance is the case of solutions with small volume and discrete counts, such as the interiors of biological cells and nanoscale engineered systems. In this case, the state variables are the integer counts of molecules, and dynamics are expressed as continuous-time, discrete-space stochastic processes \cite{gillespie1977exact,anderson2015stochastic}. 

These stochastic models are well studied.  In particular, both their stationary and finite time dynamics are of interest and have been analyzed \cite{kurtz1972relationship, anderson2015stochastic,thattai2001intrinsic, CW2016, ACK:product, challenges, multiscale, kang:separation, karlebach:modelling, grima:modelling, kurtz1976limit, agazzi:large, togashi:transitions, williams, bibbona, CJ:graphic, ACK:ACR, AEJ:ACR,ACGW2015}.  Moreover, their computational properties have been explored and it has been shown that they can simulate Turing machines, so long as a small probability of error is allowed \cite{cook2009programmability}. They are therefore capable of universal computation. However, there are known limitations on the dynamics exhibited by certain classes of stochastic CRNs \cite{lestas2008noise,lestas2010fundamental} and the full repertoire of behaviors accessible to stochastic CRNs has yet to be characterized.

A particular question is: what are the possible fluctuation sizes that can be seen in molecule counts for stochastic CRNs in stationarity? When used to model systems in thermodynamic equilibrium with particle reservoirs, i.e., models that are detailed balanced and have inflows and outflows of each species, the stationary distributions of stochastic CRNs take the form of the product of independent Poisson distributions \cite{whittle1986systems}. Hence, in these cases, the magnitude of fluctuations in the population of each molecule type is equal to the square root of the mean, since the mean and variance of Poisson distributions are equal.  In the physical sciences, it is typical to encounter systems with fluctuations with size of square root of the mean. This was stated by Schr\"{o}dinger in {\it What is life?}  when talking about the inaccuracy of physical laws and referred to it as the $\sqrt{n}$ law \cite{schrodinger1992life}. For more general detailed balanced systems, where the topology of the network restricts the set of states the process can reach, the stationary distributions continue to have product-Poisson form in the reachable space \cite{whittle1986systems}.  In fact, much more is known and the occurrence of a stationary distribution that is a product of Poissons is equivalent to the model being \textit{complex balanced}, which is a generalization of the detailed balanced condition \cite{ACK:product,CW2016}.
For other well-known  CRNs, such as models of gene regulation \cite{thattai2001intrinsic}, the variance of particle counts meets or even exceeds their mean value. 
In other examples, CRN models of low-copy-number plasmid populations in bacterial cells, where variability due to replication and partitioning could have a disruptive effect, produce distributions where the variance is less than the mean \cite{paulsson2001noise}. 
It has remained an open question whether in general there exists a bound on the variance relative to the mean for stationary distributions of stochastic CRNs,  though some results have been shown for  specific classes of models \cite{lestas2010fundamental}.

The question that we address in this paper is more general than the one posed at the beginning of the previous paragraph: is there any limit on the shape of the stationary distributions of stochastic CRNs? 
We will show that \textit{every distribution on the non-negative integer lattice can be approximated to desired precision with the stationary distribution of some stochastic CRN}. Note that, for example,  this gives an answer to the open question posed at the end of the previous paragraph: there is no  bound on the ratio between the variance and the mean of species at stationarity.

Remarkably, we will be able to approximate any distribution to any desired accuracy by restricting ourselves to two important classes of reaction networks.  The first class consists of detailed balanced models, though for this class a specific initial condition must be utilized for the desired result to hold. The second class consists of CRNs that have a unique limit distribution (i.e., the model is ergodic).   Hence, and in contrast to the first class of models, the limiting behavior of a model from this second class does not depend on the choice of initial condition.  This is a desirable feature for any model that would be physically implemented in a noisy environment.
It is interesting to note how the networks in the intersection of the two classes above have limited expressive power in terms of their limit distributions, as these can only be product of Poisson distributions \cite{whittle1986systems}.


The constructions we use are mathematically motivated, and the resulting CRNs do not necessarily have correspondence to known physical systems. For example, we sometimes make use of reactions with arbitrarily high molecularity, such as $17V\to 16V$. However, in the case of constructions utilizing detailed balanced models, we show that CRNs with a more physical interpretation can be considered, such as CRNs where the reactants and products of each reaction contain at most two molecules (see Construction~\ref{4839201483} and Remark~\ref{rmk:48398194821}). In fact, we expect that all the conclusions of our paper, namely, that the classes of detailed balanced and robust CRNs are universally approximating, will hold even if restricting the classes to suitable binary CRNs. Intuitively, this follows from the fact that the dynamics of high molecularity reactions can be approximated by the dynamics of a sequence of elementary reactions. However, this claim awaits a rigorous proof and the resulting constructions may be more complex than the ones we present here. Moreover, the number of reactions and the number of species of the constructions we describe here both scale  with the size of the support of the distribution we are approximating. Anticipating the question of descriptive complexity --- namely, can complex distributions be approximated by simple CRNs? --- we provide examples addressing it in Section~\ref{wiuiehfoweijfwKFSDFJSUHFEIRHUWEFIERHVU}. 

While our work is not the first to explore the expressive power of CRNs in terms of their limit distributions, we are the first to consider the expression of a target distribution robustly with respect to the initial condition. 
Fett, Bruck, and Riedel \cite{fett2007synthesizing} considered CRNs that make a stochastic choice among a fixed set of outcomes, with probabilities determined as a function of the counts of some input species. These CRNs require precise initial conditions, and settle to an absorbing state where no further reactions are possible. In contrast, our CRN constructions are ergodic and  therefore ``active'' for all time.  Thus, they can produce multiple observations over time, each of which has a distribution close to the limit distribution.  
Poole et al.\, \cite{poole2017chemical} show that stochastic CRNs are at least as powerful as the Boltzmann machine model from statistical machine learning in terms of their ability to generate different probability distributions. Cardelli, Kwiatkowska, and Laurenti \cite{cardelli2018programming}, like us, consider the problem of programming a CRN to approximate an arbitrary multidimensional distribution, but their constructions require precise initial conditions and are not detailed balanced. Furthermore, as is the case in \cite{fett2007synthesizing}, their CRNs are fated to a state where no reactions can take place. Plesa et al. \cite{plesa2018noise} do not consider arbitrary distributions, but develop methods for controlling noise while preserving the mean behavior of the model in a certain limit.

While our results are of independent interest as a characterization of the class of distributions that can be generated by CRNs, they may have implications for how biological cells, or engineered cell-scale molecular machines, can perform information processing in small volumes. In particular, a probability distribution can be considered as a representation of knowledge; a CRN with parameters chosen such that it generates a specific distribution can be considered to be storing said knowledge. It is therefore reasonable to ask how information stored in CRN distributions can be further processed, manipulated, and acted upon by other CRNs, or how the knowledge can be extracted by interaction with a (proto)cell's environment, in the form of tuning parameters \cite{gopalkrishnan2016scheme,virinchi2017stochastic,virinchi2018reaction}.

\section{Preliminaries}

\subsection{Notation}

Let us denote the sets of nonnegative integers, reals, and positive reals, with $\Int_{\ge0}$, $\Rel$, and $\Rel_{>0}$, respectively. In what follows, let $d\in\Int_{>0}$. For any set $K\subseteq\Rel$ of real numbers, we denote by $K^d$ the set of vectors with $d$ entries in $K$. We will often refer to the elements of $\Int_{\ge0}^d$ as {\it states}. Let $u,v\in \Rel^d$ be vectors. We write vectors as rows $v=({v(1)},\ldots,{v(d)})$, where ${v(i)}$ denotes the $i$th entry of $v$.   
For $i\in\{1,\dots,d\}$ we define $e_i$ to be the vector of $\Int^d_{\geq0}$ with  1 in the $i$th entry and $0$ otherwise, i.e. with $e_i(i)=1$, and $e_i(j)=0$, for $j\neq i$. If ${u(i)} \le {v(i)}$, for all $i\in\{1,\ldots,d\}$, we write $u\le v$. We define:
\begin{align*}
\1_{\{u\ge v\}} = \begin{cases}
1& \text{if }u\ge v\\
0& \text{ otherwise. }
\end{cases}
\end{align*}  
Let $x\in\Int_{\ge0}^d$. We define the following:
\[
u^x = \prod_{i=1}^d{u(i)}^{{x(i)}}, \quad\text{ and }\quad x! = \prod_{i=1}^d {x(i)}!,
\]
with the conventions that $0^0=1$, and $0!=1$. Let $f:\Int_{\ge0}^d\to\Rel$ be a function. We define the {\it infinity norm} as usual:
\[
  \| f \|_{\infty} =\sup_{x\in\Int_{\geq0}^{d}}\{|f(x)|\}.
\]
If $f$ satisfies $f(x)\ge 0$, for each $x\in\Int_{\ge0}^d$, and $\sum_{x\in\Int_{\ge0}^d}f(x)=1$, we say $f$ is a {\it distribution}.  We call the set $\{x\in \Int_{\ge 0}^d : f(x)>0\}$ the \textit{support} of $f$. Finally, we denote the cardinality of a set $S$ with $|S|$. 

\subsection{Model}
We are interested in the counts of molecules of different chemical species undergoing different chemical transformations. We use the standard model of stochastic chemical reaction networks in which the dynamics of the counts of the different chemical species is modeled as a continuous-time Markov chain.  We begin with the definition of a reaction network, and then characterize the dynamics.

\begin{defn}
A \emph{ chemical reaction network} (CRN) is a quadruple $\Net = (\Spc,\Cmp,\Rxn,\kappa)$ where $\Spc$, $\Cmp$, $\Rxn$, $\kappa$ are defined as follows. $\Spc$ is a finite set of \emph{species}.  $\Cmp$ is a finite set of \emph{complexes}, 
 which are linear combinations of species, with nonnegative integer coefficients. $\Rxn$ is a finite set of \emph{reactions}, that is a finite subset of $\Cmp \times\Cmp$ with the property that for any $y\in\Cmp$ we have $(y,y)\notin\Rxn$. Usually, a reaction $(y,y')$ is denoted by $y\to y'$, and we adopt this notation. Finally, given an ordering for the reactions, $\kappa$ is a vector in $\Rel_{>0}^{|\Rxn|}$ that gives every reaction a {\it rate constant}. The rate constant of reaction $y\rightarrow y'$ will be denoted by $\kappa_{y\rightarrow y'}$.
\end{defn}
Let $\Spc=\{A_1,\ldots,A_{|\Spc|}\}$ be an ordering of the species. Complexes will be regarded as vectors in $\Int_{\geq0}^{|\Spc|}$ and we use the following notation for a complex $y$:
\[
y=\sum_{i=1}^{|\Spc|} {y(i)}A_i.
\]
A network is said to be reversible if  $y'\to y \in \Rxn$ whenever $y \to y' \in \Rxn$, and we write $y \tofrom y'\in\Rxn$ for the pair of reactions. We will often summarize the sets of  complexes, reactions, and rate constants of a CRN in a {\it reaction diagram}, with reactions and their corresponding rate constants denoted in the following manner 
\begin{align*}
\sum_{i=1}^{|\Spc|}{y(i)} A_i \xrightarrow{\kappa_{y\to y'}} \sum_{i=1}^{|\Spc|}{y'(i)}A_i,&\quad \text{ for } y\to y'\in\Rxn \text{ such that }y'\to y\notin\Rxn\\
\sum_{i=1}^{|\Spc|}{y(i)} A_i \xrightleftharpoons[\kappa_{y'\to y}]{\kappa_{y\to y'}} \sum_{i=1}^{|\Spc|}{y'(i)}A_i,&\quad \text{ for } y\tofrom y'\in\Rxn.
\end{align*}

For example, consider the reaction network with species  $\Spc=\{A,B,C\}$, complexes $\Cmp=\{0, A, C, 2A + B\}$, and reactions $\Rxn=\{2A + B \to C, 0 \to A, A \to 0 \}$. Assume all rate constants are unity, i.e.\, $\kappa=(1,1,1)$. The quadruple $\Net=(\Spc,\Cmp,\Rxn,\kappa)$ forms a CRN. Moreover, the reaction network includes the reversible pair $0 \to A$ and $A\to 0$, so we can write the set of reactions as $\Rxn=\{ 2A + B \to C, 0 \tofrom A\}$. The reaction diagram of this CRN is:
\[
2A + B \xrightarrow{1} C,\qquad 0 \xrightleftharpoons[1]{1} A.
\]

The usual {\it stochastic mass action model} of a CRN  $\Net = (\Spc,\Cmp,
\Rxn,\kappa)$ is defined as  a continuous-time Markov chain (CTMC), denoted by $X(\cdot)$, where the propensity for reaction $y\rightarrow y'$ in state $x$ is given by:
\[
\lambda_{y\rightarrow y'}(x)=\kappa_{y\rightarrow y'}\frac{x!}{(x-y)!}\1_{\{x\ge y\}}.
\]
For example, the propensity of the reaction $A+B\rightarrow C$  at state $x$ is $\lambda_{A+B\rightarrow C}(x)=\kappa_{A+B\rightarrow C} x(A)x(B)$. For any pair $(x,x') \in \Int^{|\Spc|}_{\geq0}\times\Int^{|\Spc|}_{\geq0}$ the transition rate of the CTMC from $x$ to $x'$ is  given by:
\[
Q(x,x')=\sum_{\substack{\react{y}{y'}\in\Rxn\\ x'-x=y'-y}}\lambda_{\react{y}{y'}}(x).
\]

 For an initial condition $x_0\in\Int^{|\Spc|}_{\geq0}$, we use the following notation for the probability mass function:
\[
P(x,t|x_0)=P(X(t)=x|X(0)=x_0).
\]

\subsection{Key concepts}

If there is a $t>0$ for which $P(x',t | x) > 0$  we say that $x'$ is {\it reachable} from $x$. Note that if $x'$ is reachable from $x$ then it is possible to reach $x'$ from $x$ via a finite number of reactions.    The {\it reachability class} of a state $x_0$ is the set of states that are reachable from $x_0$. 
If there is a distribution $\pi(\cdot | x_0)$ for which 
\begin{equation}\label{9485739573987}
\pi(x|x_0)= \lim_{t\rightarrow\infty}P(x,t|x_0), \quad \text{ for all $x$ in the state space},
\end{equation}
then $\pi(\cdot | x_0)$ is said to be the {\it limit distribution} of the processes for initial condition $x_0$.
We note that limit distributions are stationary distributions for the models, i.e.\ distributions $\pi$ such that $P(\cdot,t)=\pi(\cdot)$ for all $t\geq0$ if $X(0)$ is distributed according to $\pi$. For a subset $\Vis\subseteq\Spc$ of species the {\it marginal} of $\pi(\cdot|x_0)$ onto $\Vis$ is:
\begin{equation}\label{768769786}
\pi_\Vis(x|x_0)= \sum_{\substack{x'\in\Int_{\geq0}^{|\Spc|}\\[0.1em](x')_\Vis=x}}\pi(x'|x_0),
\end{equation}
where $(x')_\Vis$ is the projection of $x'$ onto  the species in $\Vis$.

\begin{defn}
Let  $\Uni$ be a set of CRNs.  We say that $\Uni$ {\it approximates} a distribution  $q:\Int_{\geq0}^{d}\rightarrow[0,1]$ if for every $\varepsilon>0$ there exists a CRN $(\Spc,\Cmp,\Rxn,\kappa) \in\Uni$, an initial condition $x_0\in\Int_{\geq0}^{|\Spc|}$, and a subset of the species $\Vis \subseteq \Spc$, called the {\it visible} species, with $|\Vis| = d$, such that a limit distribution $\pi(\cdot | x_0)$ as in \eqref{9485739573987} exists and $\| \pi_\Vis(\cdot | x_0) -q \|_{\infty}<\varepsilon$. 
If $\Uni$ approximates every distribution on $\Int^m_{\ge 0}$ for every $m\ge 1$, then we say that it is {\it universally approximating}.
%
%
\end{defn}

\section{Main Results}

\subsection{Universal Approximation with Detailed Balanced Networks} 
\begin{defn}
Let the CRN $\Net = (\Spc,\Cmp,\Rxn,\kappa)$ be reversible.  If there exists a vector $c\in\Rel_{>0}^{|\Spc|}$ such that
\begin{equation}\label{584958390}
\kappa_{y\rightarrow y'}c^y = \kappa_{y'\rightarrow y}c^{y'},
\end{equation}
for each reaction $y\rightarrow y'\in\Rxn$, then we say that $\Net$ is {\it detailed balanced}. \end{defn}
\begin{rmk}
We employ one of multiple ways of defining detailed balance for CRNs. In the definition above the vector $c$ is a vector of steady-state concentrations for  the deterministic CRN with mass action kinetics that obeys the detailed balance condition \cite{horn1972general}. The choice of $c$ may not be unique.  In particular, detailed balanced networks have the remarkable property that if one positive equilibrium satisfies the detailed balanced condition \eqref{584958390}, then all positive equilibria satisfy \eqref{584958390} \cite[Theorem 3.10]{CJ:graphic}.  If one recognizes that concentrations correspond to $c_i=e^{-G_i/kT}$, where $G_i$ is the free energy of formation associated to the $i$th species, $k$ is Boltzmann's constant, and $T$ is temperature, the definition takes the form of the thermodynamic formula $\kappa_{y\to y'}/\kappa_{y'\to y}=e^{-\Delta G(y\to y')/kT}$, which relates the equilibrium constant to the change in free energy $\Delta G(y\to y') =\sum_{i=1}^{|\Spc|} G_i(y'(i)-y(i))$  associated with reaction $y\to y'$. \hfill $\triangle$
\end{rmk}

The stationary  distributions for detailed balanced models are well known:
\begin{equation}\label{3746346236238}
\pi( x| x_0) = \frac{1}{M_{ x_0}}\frac{c^ x}{ x!},
\end{equation}
for each $x$ in the reachability class of $ x_0$, where $M_{ x_0}$ is the corresponding normalization constant. This result appears in, for example, Theorem 3.2 in \cite{whittle1986systems}, or, more recently, as a special case of Theorem 4.1 in \cite{ACK:product}. However, versions of the theorem were published as early as 1958 in \cite{bartholomay1958stochastic}, and 1967 in \cite{mcquarrie1967stochastic}. 

The stationary distributions of detailed balanced systems consist exclusively of restrictions of products of Poisson distributions. Yet, as we will show, every distribution with finite support can be expressed as the marginal of the limit distribution of some detailed balanced system.

\begin{cons}[Fully connected network]\label{43294839184302143}
Let $d\geq1$ and let $q:\Int_{\geq0}^d\rightarrow[0,1]$ be a distribution with finite support $\{x_1,\dots,x_n\}$. Define the CRN $\DB(q)=(\Spc,\Cmp,\Rxn,\kappa)$  as follows. 

 \begin{itemize}
     \item The set of species is  $\Spc = \{V_1,\dots,V_d,H_1, \dots, H_m\}.$  Note that there is one species $V_i$ for each of the dimensions of $q$ and that there is one species $H_j$ for each of the points in the support of $q$.  

     \item The sets of complexes, reactions, and rate constants are given by the reaction diagram described by:
\begin{equation}\label{374290184712387432109}
H_i+\sum_{k = 1}^d {x_i(k)} V_k \xrightleftharpoons[x_i! q(x_i)]{x_j!q(x_j)} H_j + \sum_{k = 1}^d {x_j(k)} V_k, \quad
\end{equation}
 for $i,j \in \{1,\ldots,d\},\  i \ne j$. \hfill $\triangle$
 \end{itemize}

\end{cons}
States of the associated continuous-time Markov chain will reside in $\Int_{\ge0}^{d+m}$, with the first $d$ dimensions corresponding to the visible species $V_i$ and the final $m$ dimensions corresponding to the $H_j$.  We will therefore write states as  $x=(v,h)$, where $v\in\Int_{\ge0}^d$ and $h\in\Int_{\ge0}^m$. 

\begin{rmk}
Note that if the initial condition for Construction \ref{43294839184302143} is $x_0 = (x_1,e_1)$ (or, more generally, of the form $(x_i,e_i),$ for $i \in \{1,\dots,m\}$), then  at any future time there is precisely one $H_i$ molecule that has a count of one and the others have a count of zero.  Moreover, the network is designed so that if the count of $H_i$ is 1, then the counts of the $(V_1,\dots, V_d)$ are exactly $x_i.$\hfill $\triangle$
\end{rmk}

\begin{lem}\label{578943578395748390}
Let $d\ge1$ and let $q:\Int_{\ge0}^d\to[0,1]$ be a distribution with finite support. Then, $\DB(q)$ is detailed balanced and, for initial condition $x_0=(x_1,e_1)$ and set of visible species $\Vis=\{V_1,\ldots,V_d\}$, the marginal of the limit distribution satisfies $\pi_\Vis(\cdot|x_0)=q$.
\end{lem}

\begin{proof}
Let us denote the complexes of $\DB(q)$ with:
\[
	y_i = H_i + \sum_{k = 1}^d {x_i(k)} V_k,
\]
for $i\in\{1,\ldots,m\}$. Note that $c=(c^V,c^H)\in \Rel^{d + m}_{>0}$ given by
\begin{align}
\label{eq:56789765}
\begin{split}
    c^V_\ell =1,&\quad \ell\in\{1,\ldots,d\}\\
    c^H_\ell =x_\ell!q(x_\ell),&\quad\ell\in\{1,\ldots,m\}
    \end{split}
\end{align}
satisfies $\kappa_{y_i \to y_j} c^{y_i} = x_i! x_j! q(x_i) q(x_j) = \kappa_{y_j \to y_i} c^{y_j}$ for all $i,j\in\{1,\ldots,m\}$.  Hence,  $\DB(q)$ is detailed balanced. The reachability class of initial condition $x_0 = (x_1,e_1)$ is $\{(x_i,e_i)\ |\ 1\leq i\leq m\}$. Hence, for initial condition $x_0$, the limit distribution \eqref{3746346236238} satisfies:
\[
\pi((x_i,e_i)|x_0)=\frac{1}{M_{x_0}}\frac{c^{(x_i,e_i)}}{(x_i,e_i)!}=\frac{1}{M_{x_0}}\frac{c^H_i}{x_i!}=\frac{q(x_i)}{M_{x_0}}.
\]
Notice that, since $q$ is normalized, we have:
\[
M_{x_0}=\sum_{i=1}^m q(x_i)=1.
\]
Therefore, for the set of visible species $\Vis=\{V_1,\dots,V_d\}$, the marginal of the distribution satisfies:
\begin{equation}\label{6579876987685}
\pi_\Vis(x_i|x_0)=\pi(x_i,e_i|x_0) = q(x_i),
\end{equation}
for all $i\in\{1,\dots,m\},$ and it is 0 otherwise.
\end{proof}

Notice that if we modify the stoichiometry of the reactions in Construction \ref{43294839184302143}, while preserving the net change of molecules of each reaction, the reachability class for the initial condition $x_0=(x_1,e_1)$ is the same as before. Also, we can scale the rate constants of reversible pairs of reactions by the same factor without affecting the limit distribution. For example, if a state $x_i$ in the support of $q$ has exactly one more molecule of $V_k$ than another state $x_j$ in the support of $q$, we may use the reactions 
\[
H_i+V_k\xrightleftharpoons[{x_i(k)} q(x_i)]{q(x_j)} H_j,
\]
for that pair of states, instead of \eqref{374290184712387432109}. With this choice of rate constants Lemma \ref{578943578395748390} remains true. 

Let $d \ge 1$ and consider the set $\Int_{\ge0}^d$ of states. We say that two states $x,x'\in\Int_{\ge0}^d$ are {\it adjacent} if there exists $i\in\{1,\ldots,d\}$ such that $|{x(i)}-{x'(i)}|=1$, and ${x(j)}={x'(j)}$ otherwise. Let $U\subseteq\Int_{\ge0}^d$ be a set of states. If $U$ can be connected using only edges that connect adjacent states we say that $U$ is a {\it cluster}. Then, if the support of a distribution is finite and also a cluster we can modify Construction \ref{43294839184302143} to use reactions that have at most two molecules in the reactants and similarly for the products, i.e. using only bimolecular reactions.

\begin{cons}[Bimolecular indexed network]\label{4839201483}
Let $d\geq1$ and let $q:\Int_{\geq0}^d\rightarrow[0,1]$ be a distribution with finite support $\{x_1,\dots,x_n\}$. Suppose that the support is a finite cluster. Define the CRN $\BiDB(q)=(\Spc,\Cmp,\Rxn,\kappa)$  as follows. 
\begin{itemize}
    \item The set of species is the same as in Construction \ref{43294839184302143}, i.e. $\Spc=\{V_1,\ldots,V_d,H_1,\ldots,H_m\}$.
    \item The sets of complexes, reactions, and rate constants are given by the reaction diagram described by:
\[
H_i + V_k \xrightleftharpoons[{x_i(k)}  q(x_i)]{q(x_j)} H_j,
\]
for $i,j\in\{1,\ldots,m\}$, and $k\in\{1,\ldots,d\}$ such that ${x_i(k)}= x_j(k) + 1$, and the components of $x_i$ and $x_j$ are otherwise equal. That is, when $x_i - x_j = e_k.$ \hfill $\triangle$
\end{itemize}
\end{cons}
\begin{rmk}
By an argument similar to the proof of Lemma \ref{578943578395748390}, $\BiDB(q)$ is detailed balanced with detailed balanced equilibrium $c$ from \eqref{eq:56789765}.  Moreover, for initial condition $x_0=(x_1,e_1)$ and set of visible species $\Vis=\{V_1,\ldots,V_d\}$, the marginal of the limit distribution satisfies $\pi_\Vis(\cdot|x_0)=q$. \hfill $\triangle$
\end{rmk}

We can further simplify Construction \ref{4839201483} by removing reactions but preserving the connectivity of the cluster. If we remove edges of a cluster until no edge can be removed without splitting it into two disconnected graphs, the graph that results is a spanning tree. 

Let $U\subseteq\Int_{\ge0}^d$ be a cluster, and let $E\subseteq U\times U$ be a simple graph that consists only of edges that connect adjacent states. If $E$ connects all of the elements of $U$ but no proper subgraph of $E$ connects the elements of $U$ we say that $E$ is a {\it cluster spanning tree}.

\begin{cons}[Spanning tree indexed network]\label{4328492384483290}
Let $d\geq1$ and let $q:\Int_{\geq0}^d\rightarrow[0,1]$ be a distribution with finite support $U=\{x_1,\dots,x_n\}$.  Suppose that $U$ is a finite cluster. Let $E\subseteq U\times U$ be a cluster spanning tree. Define the CRN $\StiDB(q,E)=(\Spc,\Cmp,\Rxn,\kappa)$  as follows. 

 \begin{itemize}
     \item The set of species  is $\Spc=\{V_1,\ldots,V_d,H_1,\ldots,H_m\}$.
     \item The sets of complexes, reactions, and rate constants are given by the reaction diagram described by:
\[
H_i + V_k\xrightleftharpoons[{x_i(k)} q(x_i)]{q(x_j)} H_j,
\]
for $i,j\in\{1,\ldots,m\}$, and $k\in\{1,\ldots,d\}$ such that $(x_i,x_j)\in E$, and $x_i(k)= x_j(k)+1$.
Notice that $E$ consists of only edges between adjacent states so we must have that $x_i(\ell)=x_j(\ell)$ for $\ell\neq k$ since $x_i - x_j=e_k$.\hfill $\triangle$
\end{itemize}

\end{cons}
\begin{rmk}
Notice that, similar to Construction \ref{4839201483}, $\StiDB(q,E)$ is detailed balanced with detailed balanced equilibrium $c$ from \eqref{eq:56789765}.  Also, for initial condition $x_0=(x_1,e_1)$ and set of visible species $\Vis=\{V_1,\ldots,V_d\}$, the marginal of the limit distribution satisfies $\pi_\Vis(\cdot|x_0)=q$. \hfill $\triangle$
\end{rmk}

In Construction \ref{43294839184302143} the number of reactions is equal to $m(m-1)=\Ord(m^2)$, i.e. it is quadratic in the size of the support of $q$. In comparison, the number of reactions in Construction \ref{4839201483} is $\Ord(d m)$ since every point of the support has at most $2d$ adjacent states.  Finally, the number of reactions in Construction \ref{4328492384483290} is equal to $2(m-1)=\Ord(m)$, i.e. it is linear in the size of the support of $q$. 

%
\begin{figure*}
\begin{center}
\includegraphics[width=\textwidth]{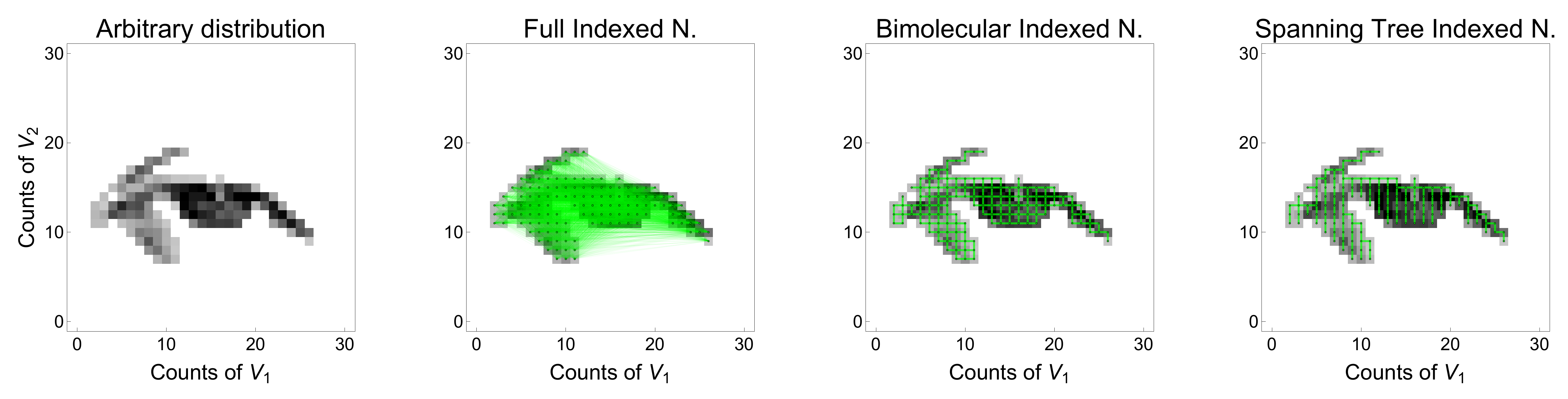}
\end{center}
\caption{Example of detailed balanced constructions. On the left we have a two-dimensional distribution with a finite cluster as support. For the Full Indexed Network every pair of elements of the support, indicated by a green edge connecting the pair, is assigned a reaction. For example, if $x_1=(5,3)$ and $x_2=(13,5)$, we have the reversible  reaction pair $H_1 + 5V_1 + 3V_2 \tofrom H_2 + 13V_1 + 5V_2$. For the Bimolecular and Spanning Tree Indexed Networks we only consider reactions between adjacent states, indicated by green edges as well. For example, we assign to the pair $x_1=(5,3),\ x_2=(5,4)$  the reversible reaction pair $H_1 \tofrom H_2 + V_2$. While the Bimolecular Indexed Network includes one reaction for every adjacent pair of elements of the support, the Spanning Tree Indexed Network includes only as many adjacent pairs as necessary to leave the support connected.}
\end{figure*}
 


\begin{lem}\label{48492318412903}
 Let $d\ge 1$, let $q:\Int_{\ge0}^d\to[0,1]$ be a distribution, and let $\varepsilon>0$. Then, there exists a distribution $q':\Int_{\ge0}^d\to[0,1]$ with finite support that satisfies $\|q-q'\|_\infty<\varepsilon.$
\end{lem}
\begin{proof}
Consider an ordering of the points in the state space $\{x_1,x_2,\ldots\}=\Int_{\ge 0}^d$, that satisfies $q(x_i) \geq q(x_j)$, whenever $i \leq j$. Since $\sum_{i = 1}^\infty q(x_i) = 1$, there is an $m$ for which $\sum_{i=1}^m q(x_i)>1-\varepsilon$.
Let $q':\Int_{\ge 0}^d \to [0,1]$ be the distribution given by:
\begin{align}\label{009875487301}
q'(x_i) = \begin{cases}
q(x_i), & i \in \{1,\ldots,m-1\}\\
\sum_{j=m}^{\infty} q(x_j), & i = m\\
0, & i \in \{m+1,\ldots\}
\end{cases}
\end{align}
Finally, notice that $q'$ satisfies:
\[
 \|q-q'\|_\infty = |q'(x_{m})-q(x_{m})|=\sum_{i=m+1}^{\infty} q(x_i) < \varepsilon,
 \]
 as desired.
\end{proof}

\begin{thm}\label{thm:dbuniversal}
The set of all detailed balanced CRNs is universally approximating. 
\end{thm}
\begin{proof}
For $d \ge 1$, let $q:\Int_{\ge 0}^d\to[0,1]$ be a distribution, and let $\varepsilon>0$. By Lemma \ref{48492318412903} we know that there exists a distribution $q':\Int_{\ge0}^d\to[0,1]$ with finite support that satisfies $\|q-q'\|<\varepsilon.$ Consider the CRN $\DB(q')$ as given in Construction \ref{43294839184302143}. We know from Lemma~\ref{578943578395748390}
 that for $\Vis=\{V_1,\ldots, V_d\}$ the marginal of the limit distribution for initial condition $x_0=(x_1,e_1)$ satisfies $\pi_\Vis(\cdot|x_0)=q'$.  Therefore, since $q$ and $\varepsilon$ were arbitrary, the set of fully connected network CRNs is universally approximating.
\end{proof}

\begin{rmk}
The proof of the theorem above can be carried out using Constructions \ref{4839201483} or \ref{4328492384483290} instead of \ref{43294839184302143}. 	If the support of a distribution is not a cluster one may connect the disconnected components using a finite number of edges and distribute a sufficiently small amount of probability throughout the states along the edges that were added. \hfill $\triangle$
\end{rmk}


In order for any of the above constructions to produce the appropriate limit distribution it is necessary to provide them with the right initial condition. Suppose instead that we wish to have a CRN whose limit distribution is independent of initial conditions. In particular, the process would be able to reach the support of its limit distribution from any state.
Moreover, if the CRN satisfies detailed balance, then a limit distribution exists \cite{ACKK:explosion} and it is a product of Poisson distributions \cite{whittle1986systems}. Therefore, such a model could not be universally approximating if it satisfies detailed balance, and a new construction is needed if one wishes to dispose of the dependence on initial conditions.

\subsection{Universal Approximation with Robust Networks}

We shift our attention to systems that have a unique limit distribution. In particular, for these systems, the limit distribution is independent of initial conditions and is robust to an arbitrary single perturbation in the counts of species at any time.

\begin{defn}
Let $\Net = (\Spc,\Cmp,\Rxn,\kappa)$ be a CRN. We say that $\Net$ is {\it robust} if the limit distribution $\pi(\cdot| x_0)$ does not depend upon $x_0$.
Since the limit distribution is unique we use $\pi(\cdot)$ and omit the initial condition.
\end{defn}
\begin{rmk}
Notice that our definition of robustness corresponds to the definition of an ergodic stochastic processes in the probability theory literature. \hfill $\triangle$
\end{rmk}

Unlike in the detailed balanced case considered in the previous section, we do not have a general form for the limit distribution of robust systems. Instead, we will provide a construction of a robust CRN that  can approximate a given point mass distribution.  We will then ``embbed'' this point mass construction in a larger robust construction that is capable of approximating an arbitrary distribution.

Let $x\in\Int_{\ge0}^d$. A {\it point mass distribution} centered at $x$, denoted by $\delta_x$, is a distribution $\delta_x:\Int_{\ge0}^d\to[0,1]$ that satisfies $\delta_x(x)=1$, and, consequently, $\delta_x(x')=0$ if $x'\ne x$. 

\begin{cons}[Point mass network]\label{constr:pmd}
Let $d \ge 1$, $x\in\Int_{\ge 0}^d$, and $\varepsilon>0$. Define  the CRN $\PMD(x,\varepsilon)=(\Spc,\Cmp,\Rxn,\kappa)$ as follows. 

\begin{itemize}
    \item The set of species is $\Spc=\{V_1,\dots,V_d\}$. 
    \item 
    The sets of complexes, reactions, and rate constants are given by the reaction diagram described by:
\begin{align*}
0 \xrightarrow{1} V_i &\quad \text { for } i\in\{1,\ldots,d\}\text{ if }x(i)\neq0,\\
({x(i)}+1)V_i \xrightarrow{2d/\varepsilon} {x(i)} V_i &\quad \text { for } i\in\{1,\ldots,d\},\\
2V_i \xrightarrow{1} 0 &\quad \text { for } i\in\{1,\ldots,d\}\text{ if }x(i)=0.
\end{align*}
\hfill $\triangle$

\end{itemize}
\end{cons}

Note that for every $i\in\{1,\dots,d\}$, exactly one of the two reactions $0\to V_i$ or $2V_i\to 0$ is present in Construction \ref{constr:pmd}.

\begin{lem}\label{lem:pmd}
Let $d \ge 1$, $x\in\Int_{\ge0}^d$, and $\varepsilon>0$. Then, $\PMD(x,\varepsilon)$ is robust, with the unique limit distribution $\pi$ satisfying $\|\pi-\delta_x\|_\infty < \varepsilon$.
\end{lem}
\begin{proof}
Notice that $\PMD(x,\varepsilon)$ consists of $d$ decoupled subnetworks, each of which controls the counts of some $V_i$ independently of all the others. We will therefore focus on one  such subnetwork and later generalize. Let $\pi_i$ be the stationary distribution of the subnetwork that keeps track of the counts of $V_i$. If $x(i)=0$, then the subsystem is
\[
V_i \xrightarrow{2d/\varepsilon} 0,\qquad 2V_i \xrightarrow{1} 0,
\]
and 
\begin{equation}\label{eq:stdistr0}
 \pi_i(0|m)=\pi_i(0)=1
\end{equation}
 for all $m\geq0$ and the result is shown.
 
 Otherwise, we have $x(i) >0$ and the subsystem is a birth-death process when restricted to the closed set $\Upsilon=\{m\in\Int\,:\,m\geq x(i)\}$.  Note that  the set $\Upsilon$ is almost surely reached from any initial condition (because of the reaction $0 \to V_i$). Hence, we can compute the unique limit distribution explicitly: for any $m,n\geq0$
\begin{equation}\label{eq:pointMass}
\pi_i(n + x(i)|m)= \pi_i(n + x(i)) = \frac{1}{M_i}\left(\varepsilon^n(2d)^{-n}\prod_{j = 1}^{n}\frac{(j-1)!}{({x(i)} + j)!}\right),
\end{equation}
with
\begin{equation}\label{eq:normalising}
M_i = \sum_{n = 0}^\infty \varepsilon^n (2d)^{-n}\prod_{j = 1}^n\frac{(j-1)!}{({x(i)} + j)!},
\end{equation}
and $\pi_i(n'|m) = 0$, for $n'\leq {x(i)}-1$.

Since the subsystems behave independently, the joint limit probability will simply be the product of the marginal limit probabilities: for any $x', x_0\in\Int_{\geq0}$
\begin{equation}\label{84938432948289580}
\pi(x'|x_0)=\pi(x')=\prod_{i=1}^d\pi_i({x'(i)}|x_0(i))=\prod_{i=1}^d\pi_i({x'(i)}).
\end{equation}

We will now show that $\|\pi-\delta_{x}\|_\infty<\varepsilon$. Since $\|p-q\|_\infty\leq 2$ for any probability distributions $p,q$, we can assume that $\varepsilon<2\leq 2d$.  Note that $|\pi(x)-\delta_{x}(x)| = 1-\pi(x)$, and for all $x'\neq x$ $|\pi(x')-\delta_{x}(x')| = \pi(x')$. Also, notice that:
\[
1-\pi(x) = \sum_{x'\neq x} \pi(x')=\sum_{x'\neq x} |\pi(x')|,
\]
so 
\begin{equation}\label{qrjjfoisjfejfoiwejfoi}
    \|\pi-\delta_{x}\|_\infty=1-\pi(x).
\end{equation}
Utilizing equation \eqref{eq:pointMass} in the case $n=0$ and \eqref{84938432948289580}, we deduce that $\pi(x) = 1/M$, where
\[
M=\prod_{\substack{1\leq i\leq d\\ x(i)\neq 0}}M_i.
\]
From \eqref{eq:normalising} we have
\[
M_i \leq \sum_{n = 0}^\infty \varepsilon^n (2d)^{-n}= \frac{1}{1-\varepsilon/(2d)},
\]
which implies
\[
M\leq\frac{1}{(1-\varepsilon/(2d))^d}.
\]
Recall that, in general, for variables $a$ and $b$ we have
\[
a^d-b^d=(a-b)(a^{d-1}+a^{d-2}b + \ldots + b^{d-1})
\]
so for $a=1$ and $b=1-\varepsilon/(2d)$ we have
\[
1-\left(1-\frac{\varepsilon}{2d}\right)^d = \frac{\varepsilon}{2d}\sum_{\ell=0}^{d-1}\left(1-\frac{\varepsilon}{2d}\right)^\ell < \frac{\varepsilon}{2},
\]
where the inequality is satisfied because $0<(1-\varepsilon/(2d))^\ell<1$ if $\varepsilon<2d$, and there are $d$ terms in the sum. Therefore, the following is true
\[
1-\pi(x) = 1-\frac{1}{M}\leq1-\left(1-\frac{\varepsilon}{2d}\right)^d <\varepsilon,
\]
which, when combined with  \eqref{qrjjfoisjfejfoiwejfoi}, concludes the proof.
\end{proof}

\begin{rmk}
Lemma \ref{lem:pmd} holds true (with an almost identical proof) if Construction \ref{constr:pmd} is replaced with 
\begin{align*}
0 \xrightarrow{1} V_i &\quad \text { for } i\in\{1,\ldots,d\},\\
({x(i)}+1)V_i \xrightarrow{2d/\varepsilon} {x(i)} V_i &\quad \text { for } i\in\{1,\ldots,d\},
\end{align*}
where a distinction between null and positive entries of $x$ is not made. However, if such a distinction is made, as in Construction \ref{constr:pmd}, then the approximation of the target limit distribution is more accurate (since the marginal limit distributions of the null entries of $x$ are \emph{exactly} the point mass distribution at 0). Moreover, the reactions of type $2V_i\to 0$ utilized when $x(i) = 0$ provide a faster convergence to the limit distribution, which will be essential to obtain the main result of this paper, which is Theorem \ref{thm:mixing} (which in turn implies Theorem \ref{thm:universal}). More details on how the convergence rate is used to prove the result are given in the Appendix.
\hfill $\triangle$
\end{rmk}


%
%
\begin{cons}[Point mass mixing network] \label{constr:mix}
Let $d\ge 1$ and $q:\Int_{\ge0}^d\to[0,1]$ be a distribution with finite support $\{x_1,x_2,\ldots,x_m\}$.  Let  $\delta>0$. Define the CRN $\PMCN(q,\delta)=(\Spc,\Cmp,\Rxn,\kappa)$  as follows. 
\begin{itemize}
    \item The set of species is $\Spc=\{V_1,\ldots,V_d,H_1,\ldots,H_m\}$. In this case, each of the species $H_1,\ldots,H_m$ will serve as a catalyst for a network that generates a point mass distribution centered at the corresponding element of the support
    \item The sets of complexes, reactions, and rate constants are given by the reaction diagram described by:
    \begin{align*}
0 \xrightleftharpoons[\delta]{\delta^2 q(x_i)} H_i,&\quad \text{ for } i\in\{1,\ldots,m\},\\
H_i \xrightarrow{1} H_i + V_j,&\quad \text{ for } i\in\{1,\ldots,m\},\text{ if }  {x_i(j)}\neq 0,
\\
H_i + ({x_i(j)}+1)V_j \xrightarrow{2d/\delta} H_i + {x_i(j)}V_j,
&\quad \text{ for } i\in\{1,\ldots,m\},\ j\in\{1,\ldots,d\},\\
H_i + 2V_j \xrightarrow{1} H_i
& \quad \text{ for } i\in\{1,\ldots,m\},\text{ if }  {x_i(j)}=0.
\end{align*}
\hfill $\triangle$
\end{itemize}
\end{cons}

\begin{figure*}
\begin{center}
\includegraphics[width=\textwidth]{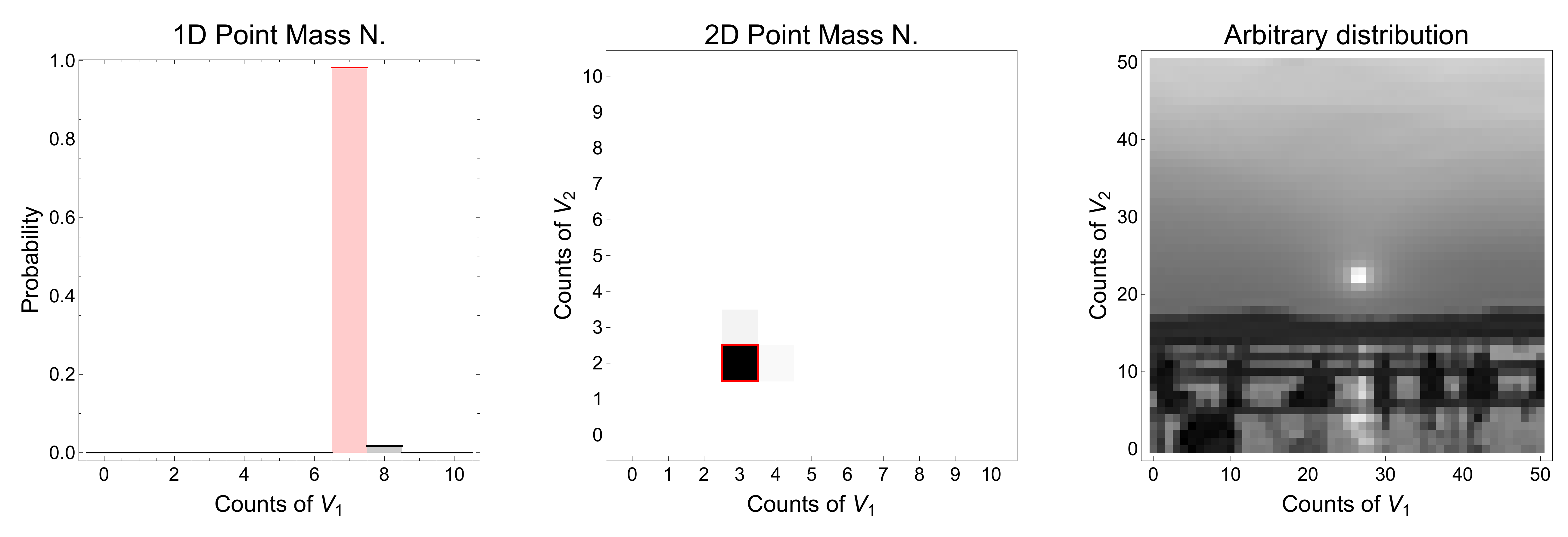}
\end{center}
\caption{Examples of robust constructions. In the first panel we see the limit distribution of Construction \ref{constr:pmd} for a 1-dimensional point mass distribution centered at $x=7$. Similarly, the second panel shows the limit distribution for Construction \ref{constr:pmd}, this time approximating a 2-dimensional point mass distribution centered at $x=(3,2)$. The third panel shows an arbitrary 2-dimensional probability distribution. Each pixel in the picture has an associated Point Mass Network as prescribed by Construction \ref{constr:mix}.}
\label{fig:robust}
\end{figure*}

\begin{rmk}
If $\delta$ is small enough, creation of catalyst species is much slower than destruction, and the probability that there is more than one catalyst species at any time can be made arbitrarily small. Furthermore, once a catalyst species is present, if the destruction rate $\delta$ is slow enough, the number of catalysts will remain unchanged long enough for the corresponding Point Mass Network to approach its limit distribution. 
\hfill $\triangle$
\end{rmk}

\begin{thm}\label{thm:mixing}
Let $d\ge 1$ and let $q:\Int_{\ge0}^d\to[0,1]$ be a distribution with finite support.   Then, (i) for every $\delta>0$ the CRN $\PMCN(q,\delta)$ is robust, and (ii) for any $\varepsilon>0$ there exists $\delta>0$ such that, for the set of visible species $\Vis=\{V_1,\ldots,V_d\}$, the marginal of the unique limit distribution of $\PMCN(q,\delta)$ satisfies $\|\pi_\Vis-q\|_\infty<\varepsilon$.
\end{thm}
The proof of Theorem~\ref{thm:mixing} is given in Appendix~\ref{sec:proof_thm_mixing}. Here we state and prove an immediate important consequence.

\begin{thm}\label{thm:universal}
The set of all robust CRNs is universally approximating.
\end{thm}
\begin{proof}
Let $d \ge 1$, $q:\Int_{\ge 0}^d\to[0,1]$ be a distribution, and let $\varepsilon>0$. By Lemma \ref{48492318412903} we know that there exists a distribution $q':\Int_{\ge0}^d\to[0,1]$ with finite support that satisfies $\|q-q'\|_\infty<\varepsilon$. Consider the CRN $\PMCN(q',\varepsilon -\|q-q'\|_\infty)$. By Theorem \ref{thm:mixing} we have that for the set of visible species $\Vis=\{V_1,\ldots,V_d\}$:
\[
\|\pi_\Vis-q'\|_\infty<\varepsilon-\|q-q'\|_\infty.
\]
Finally, by the triangle inequality we have
\[
\|\pi_\Vis-q\|_\infty \leq \|\pi_\Vis-q'\|_\infty + \|q-q'\|_\infty < \varepsilon.
\]
\end{proof}

\section{More general constructions}\label{wiuiehfoweijfwKFSDFJSUHFEIRHUWEFIERHVU}

In this section we explore generalizations of the constructions used in the previous sections. Theorem~\ref{thm:mixing} states that by following Construction~\ref{constr:mix} we can design robust CRNs whose limit distribution approximates a given distribution, with any given accuracy and for any chosen distribution. Hence, the aim of generalizing Constructions~\ref{constr:pmd} and \ref{constr:mix} does not reside in exploring a larger set of distributions to approximate, but rather in comparing different CRNs with similar stationary distributions. As an example, let $q$ be the Poisson distribution with mean $\kappa$. Using the construction provided above, we can design a robust CRN whose limit distribution is arbitrarily close to $q$, by truncating $q$ to a finite support $\Upsilon$ as in Lemma~\ref{48492318412903}, and by using Construction~\ref{constr:mix} to approximate the truncated $q$. Note that Construction~\ref{constr:mix} gives a CRN with $|\Upsilon|+1$ species, $4|\Upsilon|$ reactions, and with high values of the molecularity (i.e.\ $\max_{y\in\Cmp}\|y\|_1$) and of the logarithm of the rate constants, $|\log \kappa_{y\to y'}|$, to obtain its high accuracy. However,  the distribution $q$ can be obtained exactly as the limit distribution of the robust CRN
$$0 \xrightleftharpoons[1]{\kappa} V.$$
Hence, natural questions in terms of \emph{complexity} arise: given a distribution $q$ and a parameter $\varepsilon>0$, what is the minimal size of a robust CRN (in terms of number of species, number of reactions, and magnitude of the rate constants) whose limit distribution $\pi$ satisfies $\|\pi-q\|_\infty< \varepsilon$?

We begin by formulating a generalization of Construction~\ref{constr:pmd}. The generalization, detailed in Construction~\ref{constr:range} below, allows us to design robust CRNs whose limit distributions are arbitrarily close to the uniform distributions on $d$-dimensional boxes (see Proposition~\ref{prop:uniform}). The molecularity and the logarithm of the rate constants are similar to those of Construction~\ref{constr:mix}, but the number of species utilized is only $d$, and the number of reactions is $2d$.
\begin{defn}
 Let $d\geq1$ and let $a,b\in\Int_{\geq0}^d$ be such that $a\leq b$. We say that a probability distribution $q$ on $\Int_{\geq0}^d$ is uniform over
 $$[a,b]\doteq[a(1),b(1)]\times[a(2),b(2)]\times\cdots\times[a(d),b(d)]$$
 if for all $x\in\Int_{\geq0}^d$
 $$q(x)=\begin{cases}
 \displaystyle\prod_{i=1}^d\frac{1}{b(i)-a(i)+1} &\text{if }a\leq x\leq b\\
 0 & \text{otherwise}.
 \end{cases}$$
\end{defn}
\begin{rmk}\label{rem:pmd_is_uniform}
 Note that the point mass distribution at $x$ is a particular case of the uniform distribution, as it can be regarded as the uniform distribution over $[x,x]$.\hfill $\triangle$
\end{rmk}
\begin{cons}[Uniform distribution network]\label{constr:range}
Let $d\geq1$ and let $a,b\in\Int_{\geq0}^d$ be such that $a\leq b$. Let $\delta >0$ and define the CRN $\UD(a,b,\delta)=(\Spc,\Cmp,\Rxn,\kappa^\delta)$  as follows. 

 \begin{itemize}
     \item The set of species  is $\Spc=\{V_1,\ldots,V_d\}$.
     \item The sets of complexes, reactions, and rate constants are given by the reaction diagram described by:
  \begin{align*}
0 \xrightarrow{1} V_i &\quad \text { for } i\in\{1,\ldots,d\}\text{ if }b(i)\neq0,\\
(b(i)+1)V_i \xrightarrow{2d/\delta} {a(i)} V_i &\quad \text { for } i\in\{1,\ldots,d\},\\
2V_i \xrightarrow{1} 0 &\quad \text { for } i\in\{1,\ldots,d\}\text{ if }b(i)=0.
\end{align*}
\end{itemize}
\hfill $\triangle$
\end{cons}
\begin{rmk}
 Note that $\PMD(x,\delta)$ can be regarded as a particular case of Construction~\ref{constr:range}.  In particular,  $\PMD(x,\delta)=\UD(x,x,\delta)$ for all $x\in\Int_{\geq0}^d$ and any $\delta>0$.\hfill $\triangle$
\end{rmk}

\begin{prop}\label{prop:uniform}
Let $d\geq1$ and let $a,b\in\Int_{\geq0}^d$ be such that $a\leq b$. Then, for any choice of $\delta>0$ the CRN $\UD(a,b,\delta)$ is robust. Moreover, if we denote by $\pi^\delta$ its limit distribution and by $q$ the uniform distribution over $[a,b]$, we have
 $$\lim_{\delta\to0}\frac{\|\pi^\delta-q\|_\infty}{\delta}
 \leq 1.$$
\end{prop}
 A proof of the proposition is given in Section~\ref{dsflksdfsdlfkjsdfsdlkfjsldjjflsdjflksdjfl} in the Appendix, together with a sharper estimate on the distance between $\pi^\delta$ and $q$.  

The second, and probably more important, generalization we deal with in this section is the following. In Construction~\ref{constr:mix} we combine different robust CRNs whose limit distributions are close to point mass distributions to obtain a new robust CRN whose limit distribution is close to a mixture of point mass distributions. 
In general, given a finite number of robust CRNs with limit distributions $\pi_i$, we want to be able to design a new robust CRN that combines them, and whose limit distribution is arbitrarily close to the mixture of the distributions $\pi_i$. This can be accomplished under some general conditions, and the precise result is stated in Theorem~\ref{thm:general}. The assumptions of Theorem~\ref{thm:general} have a slightly more technical nature than those of Theorem~\ref{thm:mixing}, which is a particular case.
Note that by the known theory on detailed balanced CRNs, Lemma~\ref{lem:pmd}, and Proposition~\ref{prop:uniform}, it follows that robust CRNs whose limit distributions are arbitrarily close to a mixture of Poisson distributions, point mass distributions, and uniform distributions are readily available, and involve less species and reactions than those of Construction~\ref{constr:mix}. The precise statement of the result is given in Theorem~\ref{thm:final}.


Before stating Theorem~\ref{thm:general}, we introduce a new construction (which is a generalization of Construction~\ref{constr:mix}) and necessary concepts from the theory of stochastic processes.

\begin{cons}[Mixing network]\label{constr:general}
Let $\Fam=\{\Net_1, \dots, \Net_m\}$ be a finite ordered set of $m$ CRNs with the same set of species $\{V_1,\dots,V_d\}$. We denote by $\Rxn_i$ the set of reactions of $\Net_i$, and for each reaction $y\to y'\in \Rxn_i$ we denote by $\kappa^i_{y\to y'}$ the corresponding rate constant. Let $\zeta\in\Rel^m_{>0}$ be such that  $\sum_{i=1}^m \zeta(i)=1$, and let $\delta>0$. Define the CRN $\Mix(\Fam,\zeta,\delta)=(\Spc,\Cmp,\Rxn,\kappa)$  as follows. 
\begin{itemize}
    \item The set of species is $\Spc=\{V_1,\dots,V_d,H_1,\ldots,H_m\}$
    \item The sets of complexes, reactions, and rate constants are given by the reaction diagram described by:
\begin{align*}
0 \xrightleftharpoons[\delta]{\delta^2 \zeta(i)} H_i,&\quad \text{ for } i\in\{1,\ldots,m\},\\
H_i + y\xrightarrow{\kappa^i_{y\to y'}} H_i + y',&\quad \text{ for } i\in\{1,\ldots,m\},\  y\to y'\in \Rxn_i.
\end{align*} 
\hfill $\triangle$
\end{itemize}
\end{cons}

\begin{figure}
    \centering
    \raisebox{-0.5\height}{\includegraphics[scale=0.6]{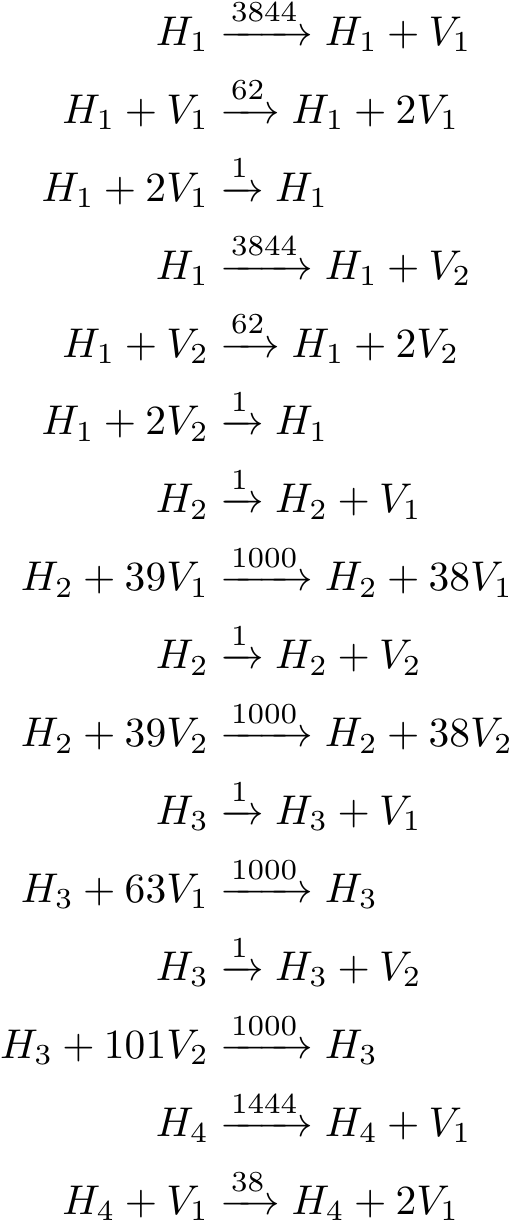}}
    \hspace{0.2in}
    \raisebox{-0.5\height}{\includegraphics[scale=0.6]{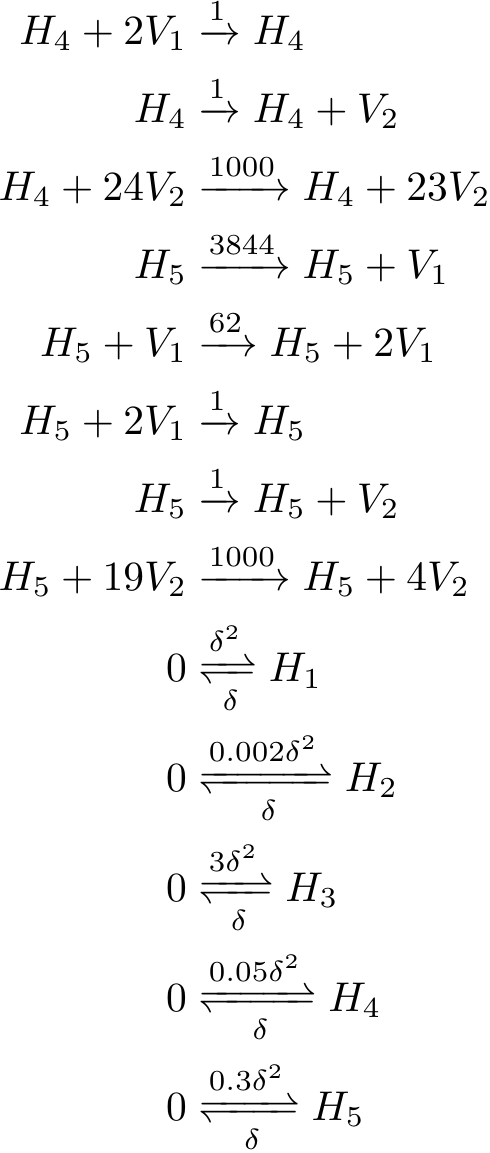}}
    \hspace{0.2in}
    \raisebox{-0.5\height}{\includegraphics[width=0.5\textwidth]{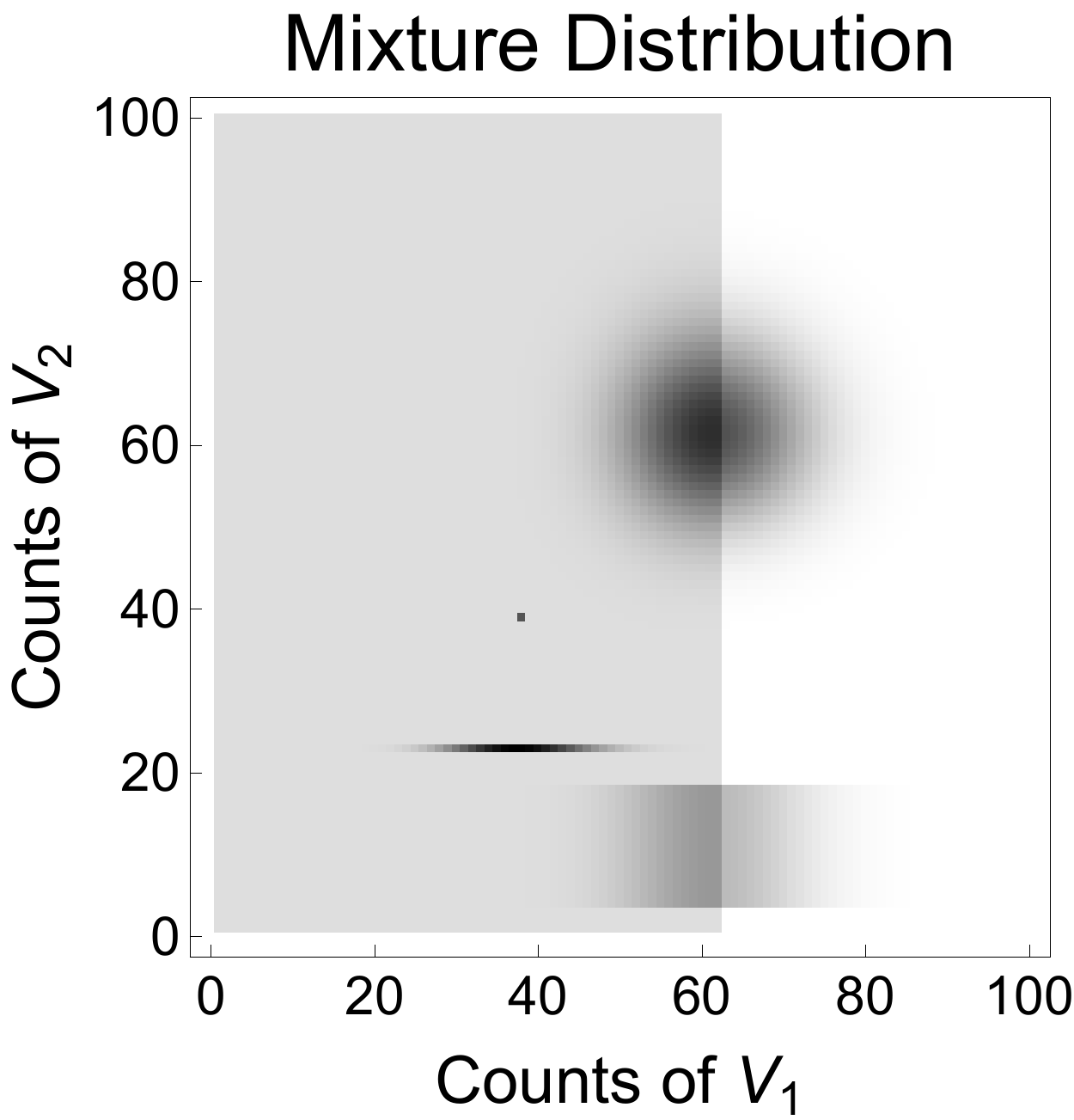}}
    \caption{Example of general mixing. The reaction diagram on the two left columns is an example of general mixing as given by Construction \ref{constr:general} with its target distribution on the right. Notice that whereas the CRN with limit distribution shown in the third panel of Figure \ref{fig:robust} has 6 reactions per pixel, giving a total of $6\times 50\times 50=1.5\times 10^4$ reactions, the above distribution is generated using only 34 reactions despite the dimension of its support being larger, namely $100\times 100$. In general, the size of the support of a distribution is not an indication of its complexity in terms of the CRNs that generate them. }
    \label{fig:mixing}
\end{figure}

\begin{defn}
 Consider a robust CRN with $d$ species, and let $\varepsilon>0$. We define the \emph{mixing time} at level $\varepsilon$ to be the quantity
 $$\tau^\varepsilon=\inf\left\{t\geq0\,:\,\sup_{x_0\in\Int^d_{\geq0}}\|P(\cdot,s|x_0)-\pi\|_\infty<\varepsilon\quad\forall s\geq t\right\}.$$
\end{defn}

The CRNs $\PMD(x,\delta)$ and $\UD(a,b,\delta)$ have finite mixing times $\tau^\varepsilon$, for any $\delta,\varepsilon>0$. This is proven in Lemma~\ref{eherhieqjfeqjfeqjfleqlf} in Section~\ref{sec:rangebis} of the Appendix.

\begin{defn}
 We say that a CRN is \emph{explosive} if there exists an initial condition $x_0$ such that for a finite time $t>0$
 $$P\left(\left.\sup_{0\leq s\leq t}\|X(s)\|_\infty=\infty \right| x_0\right)>0.$$
 We say that the CRN is \emph{non-explosive} otherwise.
\end{defn}

\begin{thm}\label{thm:general}
 Let $\Fam=\{\Net_1, \dots, \Net_m\}$ be a finite ordered set of $m$ CRNs with the same set of species $\Vis=\{V_1,\dots,V_d\}$. Assume that each CRN $\Net_i$ is robust, denote by $\pi_i$ its limit distribution and by $\tau^\varepsilon_i$ its mixing time at level $\varepsilon>0$. Let $\zeta\in\Rel^m_{>0}$ be such that $\sum_{i=1}^m \zeta(i)=1$, and assume that for every $\varepsilon>0$
 $$\max_{1\leq i\leq m}\tau^\varepsilon_i<\infty.$$
 Moreover, assume that for every $\delta>0$ the CRN $\Mix(\Fam, \zeta, \delta)$ is non-explosive. Then, for every $\delta>0$ the CRN $\Mix(\Fam, \zeta, \delta)$ is robust, and, if we denote by $\pi^\delta$ the limit distribution of $\Mix(\Fam, \zeta, \delta)$, we have
 $$\lim_{\delta\to 0}\left\|\pi^\delta_\Vis-\sum_{i=1}^m \zeta(i)\pi_i\right\|_{\infty}=0.$$
\end{thm}

In order to choose $\delta$ such that the distance
$$\left\|\pi^\delta_\Vis-\sum_{i=1}^m \zeta(i)\pi_i\right\|_{\infty}$$
is smaller than a given quantity, it is important to have upper bounds on the mixing times $\tau^\varepsilon_i$. In Appendix~\ref{sec:bounds}, such bounds are developed for the CRN $\PMD(x,\delta)$. We have seen in Theorem~\ref{thm:mixing} (which is a consequence of Theorem~\ref{thm:general}) how a family of point mass networks can be used to construct a robust CRN whose limit distribution is arbitrarily close to a given distribution $q$.

However, the application of Theorem~\ref{thm:general} does not need to be limited to point mass networks. Before stating the next result, which is more general than Theorem~\ref{thm:mixing}, we introduce a choice of CRN constructions with a product-form Poisson as limit distribution. Note that, due to \cite{ACK:product}, many different choices are possible. Here we choose one with a finite mixing time, so that Theorem~\ref{thm:general} can be applied.

\begin{cons}(Product-form Poisson Network)\label{constr:poisson}
 Let $c\in\Rel_{>0}^d$. Define the CRN $\PPN(c)=(\Spc, \Cmp, \Rxn, \kappa)$ as follows.
 \begin{itemize}
     \item The set of species is $\Spc=\{V_1,\dots,V_d\}$.
     \item The sets of complexes, reactions, and rate constants are given by the reaction diagram described by: 
     \begin{align*}
0 \xrightarrow{c(i)^2} V_i &\quad \text { for } i\in\{1,\ldots,d\}\\
V_i \xrightarrow{c(i)} 2V_i &\quad \text { for } i\in\{1,\ldots,d\},\\
2V_i \xrightarrow{1} 0 &\quad \text { for } i\in\{1,\ldots,d\}.
 \end{align*}
 \end{itemize}
 \hfill $\triangle$
\end{cons}

\begin{thm}\label{thm:final}
Let $\{\pi_1, \dots, \pi_m\}$ be a family of distributions on $\Int_{\geq0}^d$ such that there is a partition $\{I_1, I_2, I_3\}$ of $\{1,\dots,m\}$ satisfying the following:
\begin{itemize}
    \item for all $i\in I_1$, $\pi_i$ is a point mass distribution at some $x_i\in\Int_{\geq0}^d$;
    \item for all $i\in I_2$, $\pi_i$ is a uniform distribution over $[a(i),b(i)]$ for some $a(i)\leq b(i)\in\Int_{\geq0}^d$;
    \item for all $i\in I_3$, $\pi_i$ is a product-form Poisson distribution with mean $c_i\in\Rel_{>0}^d$.
\end{itemize}
Let $\Fam^\delta=\{\Net^\delta_1, \dots, \Net^\delta_m\}$ be a family of reaction networks with common set of species $\Vis=\{V_1,\dots,V_d\}$, such that
\begin{itemize}
    \item for all $i\in I_1$, $\Net^\delta_i=\PMD(x_i,\delta)$;
    \item for all $i\in I_2$, $\Net^\delta_i=\UD(a(i),b(i),\delta)$;
    \item for all $i\in I_3$, $\Net^\delta_i=\Net_i=\PPN(c_i)$.
\end{itemize}
Let $\zeta\in\Rel^m_{>0}$ with $\sum_{i=1}^m\zeta(i)=1$. Then, for any $\delta>0$ the CRN $\Mix(\Fam^\delta, \zeta, \delta)$ is robust. Moreover, if $\pi^\delta$ denotes its limit distribution, then
$$\lim_{\delta\to 0}\left\|\pi^\delta_\Vis-\sum_{i=1}^m \zeta(i)\pi_i\right\|_\infty=0.$$

\end{thm}
 The proofs of Theorem~\ref{thm:general} and Theorem~\ref{thm:final} are given in Sections~\ref{sec:proof_thm_general} and \ref{sec:proof_thm_final} of the Appendix, respectively.
 
As a final remark, note that in Constructions~\ref{constr:pmd}, \ref{constr:range}, and \ref{constr:poisson}, the species do no interact (i.e.\ two different species are never involved in the same reaction). As a consequence, the counts of the different species evolve independently. It is therefore straightforward to obtain CRNs whose limit distribution approximates products of Poisson distributions and uniform distributions. It is indeed sufficient to consider different constructions for different species: consider for example the distribution $q$ on $\Int_{\geq0}^2$ given by $q(v_1, v_2)=q_1(v_1)q_2(v_2)$, where $q_1$ is uniform over $\{a, a+1, \dots, b\}$ (say with $1\leq a\leq b$) and $q_2$ is Poisson with mean $c$. Then, $q$ is approximated by the limit distribution of
 \begin{gather*}
  0 \xrightarrow{1} V_1\\
(b+1)V_1 \xrightarrow{2/\delta} aV_1\\
 0 \xrightarrow{c^2} V_2\\
V_2 \xrightarrow{c} 2V_2\\
2V_2 \xrightarrow{1} 0.
 \end{gather*}
 The design of the CRN in Figure~\ref{fig:robust} follows this idea.
 
 \section*{Acknowledgment}
 The current structure of the project was conceived when the four authors participated in the BIRS 5-day Workshop ``Mathematical Analysis of Biological Interaction Networks.'' Parts of the proofs for the present work were then completed while two of the authors were taking part in the AIM SQuaRE workshop ``Dynamical properties of deterministic and stochastic models of reaction networks.'' We thank the Banff International Research Station and the American Institute of Mathematics for making this possible.
 
 Anderson gratefully acknowledges support via the Army Research Office through grant  W911NF-18-1-0324. Winfree gratefully acknowledges support via the National Science Foundation ``Expedition in Computing'' grant CCF-1317694.
 
 This material is based upon work supported by the National Science Foundation Graduate Research Fellowship Program under Grant No. DGE‐1745301. Any opinions, findings, and conclusions or recommendations expressed in this material are those of the author(s) and do not necessarily reflect the views of the National Science Foundation.
 
	\bibliographystyle{plain} 
	\bibliography{mybib}

\newpage
\begin{center}
    {\bfseries\Large Appendix}
\end{center}
\appendix
\renewcommand{\thesection}{\Alph{section}}
\section{Proofs and estimates}

The aim of this section is to provide a complete proof of Proposition~\ref{prop:uniform}, Theorem \ref{thm:mixing}, Theorem \ref{thm:general}, and Theorem~\ref{thm:final}. In particular, we will show how Theorem~\ref{thm:mixing} follows from Theorem~\ref{thm:final}, which in turn follows from Theorem~\ref{thm:general}. 

From the proof of Theorem~\ref{thm:general} it will emerge how, for a fixed $\varepsilon>0$, the choice of $\delta$ such that 
$$\left\|\pi^\delta_\Vis-\sum_{i=1}^m \zeta(i)\pi_i\right\|_{\infty}<\varepsilon$$
depends on the mixing times of the reaction networks $\Net_1, \dots, \Net_m$. This holds in the particular case of Theorem~\ref{thm:mixing} as well.
Hence, to guide the design of networks that approximate a given distribution with accuracy $\varepsilon$, we will provide in Appendix~\ref{sec:bounds} useful estimates on the mixing times of $\PMD(x,\delta)$.

\subsection{Proof of Theorem~\ref{thm:general}}\label{sec:proof_thm_general}

 Denote by $X^\delta(\cdot)$ the continuous-time Markov chain associated with $\Mix(\Fam, \zeta, \delta)$. Define $X_\Vis^\delta(\cdot)$ and $X_\Hsp^\delta(\cdot)$ as the projections of $X^\delta(\cdot)$ onto the components of the species in $\Vis=\{V_1,\dots,V_d\}$ and $\Hsp=\{H_1,\dots,H_m\}$, respectively. 
Moreover, for convenience we will write a state of $\Mix(\Fam,\zeta,\delta)$ as
$$(v,h)\in\Int_{\geq0}^{d+m},$$
with $v$ and $h$ indicating the components of the species in $\Vis$ and $\Hsp$, respectively.
%
 
 Note that $X^\delta_\Hsp(\cdot)$ is a continuous-time Markov chain itself, and is distributed according to the subnetwork of $\text{mix}(\Fam,\zeta,\delta)$, which we will denote by $\Net_\Hsp=(\Hsp, \Cmp_\Hsp,\Rxn_\Hsp)$, given by
  \begin{equation}\label{eq:subnetwork}
   0\xrightleftharpoons[\delta]{\zeta(i)\delta^2}H_i\quad \text{for }1\leq i\leq m.
  \end{equation}
 $\Net_\Hsp$ is detailed balanced with detailed balanced equilibrium $\delta\zeta$, and its state space is irreducible. It follows that $\Net_\Hsp$ admits a unique stationary distribution $\pi^\delta_{\Hsp}$ defined by
 \begin{equation}\label{eq:st_dist_H}
     \pi^\delta_{\Hsp}(h)=e^{-\delta}\prod_{i=1}^{m} \frac{(\delta\zeta(i))^{h_i}}{h_i!}
 \end{equation}
 for all $h\in\Int_{\geq0}^m$.
 
 Consider a vector $v\in\Int^d_{\geq0}$ that is  a positive recurrent  state for at least one reaction system $\Net_i$, $i \in \{1,\dots, d\}$. 
  Moreover, denote by $\sigma^\delta(v)$ the time of the first visit of $X^\delta(\cdot)$ to $(v,0)$, defined as 
 $$\sigma^\delta(v)=\inf\{t\geq 0\,:\, X^\delta(t)=(v,0)\text{ and }X^\delta(s)\neq(v,0)\text{ for some } 0\leq s<t\}.$$
 We prove that $\Mix(\Fam, \zeta, \delta)$ is robust by proving that for any $(v_0,h_0)\in\Int^{d+m}_{\geq0}$
 \begin{equation}\label{eq:finite_expectation}
  E[\sigma^\delta(v)|X^\delta(0)=(v_0,h_0)]<\infty.
 \end{equation}
 Indeed, if \eqref{eq:finite_expectation} holds, then $(v,0)$ is positive recurrent by definition, hence there exists a stationary distribution $\pi^\delta$ whose support coincides with the closed irreducible component that contains $(v,0)$, and such stationary distribution is unique. Moreover, a unique closed irreducible set exists and it is eventually reached with probability 1  from all states of $\Int_{\geq0}^{d+m}$, otherwise \eqref{eq:finite_expectation} could not hold.
 
 We need to prove \eqref{eq:finite_expectation}. 
 Let $e_i\in\Int_{\geq0}^m$ be the $i$th vector of the canonical basis, namely the vector with 1 in the $i$th entry and 0 in the other components. We have that \eqref{eq:subnetwork} is positive recurrent and irreducible, and that $X^\delta(\cdot)$ is non-explosive and therefore well-defined for all times greater than 0. Hence, given $X^\delta(0)=(v_0,h_0)$, the chain will satisfy $X_\Hsp^\delta(t)=e_i$ after a time with finite expectation. Assume $X_\Hsp^\delta(t)=e_i$, and let $u$ be the time  until the next change in copy-numbers of the species $\{H_1,\dots,H_m\}$. Then, $u$ is exponentially distributed with rate rate $\delta+\delta^2$, independently of the value of $X_\Vis^\delta(t)$. It follows that there is a positive probability $\varphi_i(\delta,v)$ that $u>\tau^{\pi_i(v)/2}$, where
 $$\tau^{\eta}=\max_{1\leq i\leq m}\tau^{\eta}_i$$
 is finite by assumption for all $\eta>0$. Moreover, by definition of mixing times,
 $$P\Big(X^\delta(t+u)=(v,0)\,\Big|\,X_\Hsp^\delta(t)=e_i, u>\tau^{\pi_i(v)/2}, X_\Vis^\delta(t)=v'\Big)\geq \frac{\pi_i(v)}{2}>0,$$
 independently of $v'\in\Int_{\geq0}^d$.
 Hence, the number of times the chain satisfies $X_\Hsp^{\delta}(t)=e_i$ before visiting $(v,0)$ is stochastically bounded from above by a geometric random variable with mean $\Big(\varphi_i(\delta, v) \pi_i(v)/2\Big)^{-1}$. Moreover, the expected time between two visits of $\eqref{eq:subnetwork}$ to $e_i$ is finite. Hence, \eqref{eq:finite_expectation} holds.
 
  For all $\delta>0$ we have
 $$\pi^\delta_\Vis(v)-\pi^\delta(v,0)=\sum_{h\in\Int_{\geq0}^{m}\setminus\{0\}}\pi^\delta(v,h),$$
 which implies
 \begin{equation}\label{347593847598347598ekjvnskjfheiufy}
 0\leq \pi^\delta_\Vis(v)-\pi^\delta(v,0)\leq \sum_{h\in\Int_{\geq0}^{m}\setminus\{0\}}\pi^\delta_\Hsp(h)=1-e^{-\delta}.
 \end{equation}
 
   For any $v\in \Int_{\geq0}^{d}$, let $Y^\delta_v(t)$ be the time $X^\delta(\cdot)$ spends in state $(v,0)\in\Int_{\geq0}^{d+m}$ by time $t$, that is
 $$Y^\delta_v(t)=\int_0^t \mathbbm{1}_{\{(v,0)\}}(X^\delta(s))ds.$$
 By classical Markov chain theory and by robustness of $\Mix(\Fam, \zeta, \delta)$ we have that
 $$\lim_{t\to\infty}\frac{Y^\delta_v(t)}{t}=\pi^\delta(v,0)$$
 almost surely, for any initial condition $X^\delta(0)$.
 
 We now assume that for any $v\in \Int_{\geq0}^{d}$ and any $\varepsilon>0$, there exists $\delta^{\varepsilon,v}$ such that if $\delta\leq \delta^{\varepsilon,v}$ then
 \begin{equation}\label{eq:convergence}
  \left|\lim_{t\to\infty}\frac{Y^\delta_v(t)}{t}-\sum_{i=1}^m \zeta(i)\pi_i(v)\right|<\frac{3}{4}\varepsilon
 \end{equation}
 almost surely, independently on the initial condition $X^\delta(0)$. 
 For $\delta$ small enough, both $\delta\leq\delta^{\varepsilon,v}$ and $1-e^{-\delta}<\varepsilon/4$ hold. Hence, by \eqref{347593847598347598ekjvnskjfheiufy}, \eqref{eq:convergence}, and the triangular inequality
 $$\left|\pi^\delta_\Vis(v)-\sum_{i=1}^m \zeta(i)\pi_i(v)\right|<\varepsilon.$$
 For any $\varepsilon>0$, there exists a compact set $K^\varepsilon\subset\Int_{\geq0}^{d}$ such that for any $v\notin K^\varepsilon$
 $$\pi^\delta_\Vis(v)+\sum_{i=1}^m \zeta(i)\pi_i(v)< \varepsilon.$$
 Since $K^\varepsilon$ is compact, the minimum $\delta^\varepsilon=\min_{v\in K^\varepsilon} \delta^{\varepsilon, v}$ exists and is positive. 
 Hence, for all $v\in\Int_{\geq0}^{d}$ and for all $\delta$ small enough such that $\delta\leq\delta^{\varepsilon}$ and $1-e^{-\delta}<\varepsilon/4$, we have
 $$\left|\pi^\delta_\Vis(v)-\sum_{i=1}^m \zeta(i)\pi_i(v)\right|<\begin{cases}
                                                                             \varepsilon & \text{if }v\in K^{\varepsilon}\\
                                                                             \pi^\delta_\Vis(v)+\sum_{i=1}^m \zeta(i)\pi_i(v)< \varepsilon& \text{if }v\notin K^{\varepsilon}
                                                                            \end{cases}$$
 and the proof is concluded. Hence, it suffices to show \eqref{eq:convergence}.

  Let $t_0=0$ and define recursively 
 $$t_j=\inf\{t\geq t_{j-1}\,:\, X_\Hsp^\delta(t)=0\text{ and }X_\Hsp^\delta(s)\neq0\text{ for some } t_{j-1}<s<t\}$$
 for $j\geq1$. That is, $t_j$ with $j\geq 1$ is the time of the $j$th visit to a state with no molecules of species $H_1, \dots, H_m$.
 
 For any $j\geq1$, let $s_j$ denote the holding time in the state with no molecules of $\{H_1,\dots,H_m\}$, measured from time $t_j$. Then, $s_j$ is exponentially distributed with rate
 $$\sum_{y\to y'\in\Rxn_\Hsp}\lambda_{y\to y'}(0)=\sum_{i=1}^m\delta^2\zeta(i)=\delta^2.$$
It follows from classical renewal theory and from \eqref{eq:st_dist_H} that for any $j\geq1$ 
 $$E[t_{j+1}-t_{j}]=\frac{1}{\pi^\delta_\Hsp(0)\sum_{y\to y'\in\Rxn_\Hsp}\lambda_{y\to y'}(0)}=\frac{e^{\delta}}{\delta^2}.$$
 It follows that, with probability 1, $\lim_{j\to\infty}t_j=\infty$. This in turn implies that almost surely
 $$\lim_{t\to\infty}\frac{Y^\delta_v(t)}{t}=\lim_{j\to\infty}\frac{Y^\delta_v(t_j)}{t_j}.$$

 For all $j\geq1$, independently of the value of $X_\Vis^\delta(t_j)$, a molecule of $H_i$ is produced at time $t_j+s_j$ with probability $\zeta(i)$. Let $u_j$ denote the molecule lifetime. Note that with probability
 $$\frac{\delta}{\delta + \delta^2}=\frac{1}{1+\delta}$$
 the molecule of $H_i$ is degraded before another molecule of a species in $\Hsp$ is produced. For convenience, denote this event by $A_j$. Given that $A_j$ occurs, $u_j$ is the minimum between the degradation of the $H_i$ molecule (exponentially distributed with rate $\delta$) and the time until the production of another molecule of a species in $\Hsp$ (exponentially distributed with rate $\delta^2$). Hence, given that $A_j$ occurs, $u_j$ is exponentially distributed with rate $\delta + \delta^2$. It follows that
 $$P(u_j>\tau^{\varepsilon/2}|A_j)=e^{-(\delta + \delta^2)\tau^{\varepsilon/2}}.$$
 Given that $u_j>\tau^{\varepsilon/2}$, all the reactions of the system $\Net_i$ can take place for a time longer than $\tau^{\varepsilon/2}$, and these are the only reactions that can occur. Thus, by the definition of mixing times,
 $$\pi_i(v)-\frac{\varepsilon}{2}\leq P\Big(X_\Vis^\delta(t_{j+1})=v\,\Big|\,X_\Hsp^\delta(t_j+s_j)=H_i, A_j, u_j>\tau^{\varepsilon/2}, X_\Vis^\delta(t_j)=v'\Big)\leq \pi_i(v)+\frac{\varepsilon}{2}$$
 for all $v, v'\in\Int_{\geq0}^d$. Note that the bounds do not depend on $X_\Vis^\delta(t_j)=v'$. In conclusion, by conditioning and by using the probabilities of the conditioning events calculated above, we obtain
 \begin{align}\label{lower_bound}
  P\Big(X_\Vis^\delta(t_{j+1})=v|X_\Vis^\delta(t_j)=v'\Big)\geq& \sum_{i=1}^m\left(\pi_i(v)-\frac{\varepsilon}{2}\right)\zeta(i)\frac{1}{1+\delta}e^{-(\delta + \delta^2)\tau^{\varepsilon/2}}\notag\\
  \doteq& b^\delta(v)\\
  \label{upper_bound}
  P\Big(X_\Vis^\delta(t_{j+1})=v|X_\Vis^\delta(t_j)=v'\Big)\leq& \sum_{i=1}^m\left(\pi_i(v)+\frac{\varepsilon}{2}\right)\zeta(i)\frac{1}{1+\delta}e^{-(\delta + \delta^2)\tau^{\varepsilon/2}}\notag\\
  &+\frac{\delta}{1+\delta}+\frac{1}{1+\delta}\left(1-e^{-(\delta + \delta^2)\tau^{\varepsilon/2}}\right)\notag\\
  \doteq& B^\delta(v).
 \end{align}

  The sequence $D^\delta(j)=X_\Vis^\delta(t_j)$ for $j\in\Int_{\geq0}$ defines a discrete time Markov chain. Since $X^\delta(\cdot)$ is robust, $D^\delta(\cdot)$ has a unique closed irreducible set $\Upsilon$, and for any $v'\in\Int^d_{\geq0}$ we have
  $$\lim_{j\to\infty}P(D^\delta(j)\in\Upsilon\,|\,D^\delta(0)=v')=1.$$
  Moreover, \eqref{lower_bound} and \eqref{upper_bound} give a lower and upper bound on the transition probabilities to a state $v$ from a state $v'$, which does not depend on $v'$. Hence, for small enough $\varepsilon$ and small enough $\delta$ such that $b^\delta(v)>0$, $D^\delta(\cdot)$ restricted to $\Upsilon$ is aperiodic and positive recurrent. It follows that there exists a limit distribution $\gamma$ such that for any $v, v'\in\Int_{\geq0}^d$
  $$\lim_{j\to\infty}P(D^\delta(j)=v\,|\,D^\delta(0)=v')=\gamma(v),$$
  independently of $v'$. Furthermore, since the argument of the limit is bounded from below by $b^\delta(v)$ and from above by $B^\delta(v)$, we have
  $$b^\delta(v)\leq \gamma(v)\leq B^\delta(v).$$
  If $W_j(v)$ is the number of visits of $D^\delta(\cdot)$ to $v$ up to step $j$ (included), we have that with probability 1 $\lim_{j\to\infty}W_j(v)=\infty$, and
  $$\lim_{j\to\infty}\frac{W_j(v)}{j}=\gamma(v).$$

  By the strong law of large numbers, we have that almost surely
  \begin{align*}
   \lim_{j\to\infty}\frac{Y^\delta_v(t_j)}{t_j}&=\lim_{j\to\infty}\frac{Y^\delta_v(t_1)}{t_j}+\lim_{j\to\infty}\frac{t_j-t_1}{t_j}\cdot\frac{Y^\delta_v(t_j)-Y^\delta_v(t_1)}{t_j-t_1}\\
   &=0+\lim_{j\to\infty}\frac{t_j-t_1}{t_j}\cdot\frac{\sum_{i=1}^{W_{j-1}(v)}s_i}{\sum_{i=2}^j (t_i-t_{i-1})}\\
   &=\lim_{j\to\infty}\frac{t_j-t_1}{t_j}\cdot\frac{\sum_{i=1}^{W_{j-1}(v)}s_i}{W_{j-1}(v)}\cdot\frac{W_{j-1}(v)}{j-1}\cdot\frac{j-1}{\sum_{i=2}^j (t_i-t_{i-1})}\\
   &=1\cdot \frac{1}{\delta^2} \cdot \gamma(x) \cdot \frac{\delta^2}{e^\delta}=\gamma(x)e^{-\delta}.
  \end{align*}

Hence,
  $$e^{-\delta}b^\delta(v)\leq \lim_{j\to\infty}\frac{Y^\delta_v(t_j)}{t_j}\leq e^{-\delta}B^\delta(v).$$
  If $\delta$ is small enough,
  \begin{align*}
   e^{-\delta}b^\delta(v)\geq \sum_{i=1}^{m}\Big(\zeta(i)\pi_i(v)\Big)-\frac{3}{4}\varepsilon\\
   e^{-\delta}B^\delta(v)\leq \sum_{i=1}^{m}\Big(\zeta(i)\pi_i(v)\Big)+\frac{3}{4}\varepsilon,
  \end{align*}
  which proves \eqref{eq:convergence} and concludes the proof. \hfill\qed


\subsection{Analysis of Constructions~\ref{constr:pmd} and \ref{constr:range}}\label{sec:rangebis}

In this section, we study the limit distributions and the mixing times of a construction that generalize slightly Constructions~\ref{constr:pmd} and \ref{constr:range}, by allowing for a more general choice of rate constants. We also give explicit bounds on the distance between the uniform distribution and the limit distribution of the construction presented here. The bounds provided here are in general sharper than those presented in the main text.

\begin{consbis}{constr:range}[General uniform distribution network]\label{constr:rangebis}
Let $d\geq1$ and let $a,b\in\Int_{\geq0}^d$ be such that $a\leq b$. Define the CRN $\UDbis(a,b,\kappa)=(\Spc,\Cmp,\Rxn,\kappa)$  as follows. 

 \begin{itemize}
     \item The set of species  is $\Spc=\{V_1,\ldots,V_d\}$.
     \item The sets of complexes, reactions, and rate constants are given by the reaction diagram described below:
  \begin{align*}
0 \xrightarrow{\kappa^i_1} V_i &\quad \text { for } i\in\{1,\ldots,d\}\text{ if }b(i)\neq0,\\
(b(i)+1)V_i \xrightarrow{\kappa^i_2} {a(i)} V_i &\quad \text { for } i\in\{1,\ldots,d\},\\
2V_i \xrightarrow{\kappa^i_3} 0 &\quad \text { for } i\in\{1,\ldots,d\}\text{ if }b(i)=0.
\end{align*}
\end{itemize}
\hfill $\triangle$
\end{consbis}

In what follows, we will denote by $X^\kappa(\cdot)$ the continuous-time Markov chain associated with the CRN $\UDbis(a,b,\kappa)$, and we will use the notation $r_i=b(i)-a(i)+1$ for all $i\in\{1,\ldots,d\}$. Further, we will denote by $X^\kappa(\cdot,i)$ the $i$th component of $X^\kappa(\cdot)$. Note that the components $X^\kappa(\cdot,i)$ are distributed as independent continuous-time Markov chains. In particular, $X^\kappa(\cdot,i)$ is distributed as the process associated with the subnetwork of $\UDbis(a,b,\kappa)$ given by the reactions changing the species $V_i$. We will denote by $A^{\kappa, i}$ the generator of $X^\kappa(\cdot,i)$ (see \cite{ethierkurtz}). We define $\sigma^{\kappa,i}$ and $\sigma^{\kappa,i}_{v(i)}$ as the hitting times of $[a(i),b(i)]$ and of $v(i)\in\Int_{\geq0}$, respectively:
\begin{align*}
    \sigma^{\kappa,i}&=\inf\{t>0:X^\kappa(t,i)\in[a(i),b(i)]\},\\
    \sigma^{\kappa,i}_{v(i)}&=\min\{t>0:X^\kappa(t,i)=v(i)\text{ and }X^\kappa(s,i)\neq v(i)\text{ for some }s<t\}.
\end{align*}
Finally, we denote by $E^{\kappa,i}_{v_0(i)}[\cdot]$ the expectation with respect to the distribution of $X^\kappa(\cdot,i)$ given $X^\kappa(0,i)=v_0(i)$

\begin{alem}\label{lem:superrangebis}
  Let $d\geq1$ and let $a,b\in\Int_{\geq0}^d$ be such that $a\leq b$. Consider the function $L(\cdot)$, defined as $L(v(i))=v(i)$ for all $v(i)\in \Int_{\geq0}$. Then, $\lim_{v(i)\to\infty}L(v(i))=\infty$. Moreover, for any $\kappa>0$ and any $i\in\{1,\dots,d\}$, there exists $\alpha_i\in\Rel_{>0}$ and a compact set $K_i\subset \Int_{\geq 0}$ such that
  $$A^{\kappa, i}L(v(i))\leq -\alpha_i L^2(v(i))\quad\text{for all }v(i)\notin K_i.$$
\end{alem}
\begin{proof}
 Clearly, $\lim_{v(i)\to\infty}L(v(i))=\infty$. Now assume that $b(i)=0$. Then, for all $v(i)\geq2$ we have
 $$A^{\kappa, i}L(v(i))=-\kappa^i_2v(i)-2\kappa^i_3v(i)(v(i)-1),$$
 which implies that for $v(i)$ big enough $A^{\kappa, i}L(v(i))\leq -\kappa^i_3v(i)^2=-\kappa^i_3 L^2(v(i))$, and the result holds.
 
 Assume that $b(i)\neq0$. For all $v(i)\geq b(i)+1$ we have
 $$A^{\kappa, i}L(v(i))=\kappa^i_1-\kappa^i_2\frac{v(i)!}{(v(i)-b(i)-1)!},$$
 which is smaller than or equal to $-\frac{\kappa^i_2}{2} L^2(v(i))$ for $v(i)$ large enough. The proof is then concluded.
\end{proof}

\begin{alem}\label{eherhieqjfeqjfeqjfleqlf}
 Let $d\geq1$ and let $a,b\in\Int_{\geq0}^d$ be such that $a\leq b$. For any $\kappa>0$, $\UDbis(a,b,\kappa)$ is robust and the support of its limit distribution is
 $$\Theta=\{v\in\Int_{\geq0}^d:v\geq a\quad\text{and}\quad v(i)=0\quad\text{if}\quad b(i)=0\}=\left(\bigtimes_{\substack{1\leq i\leq d\\ b(i)=0}}\{0\}\right)\times\left(\bigtimes_{\substack{1\leq i\leq d\\ b(i)\neq0}}\{v(i)\,:\,v(i)\geq a(i)\}\right).$$
 Moreover, for any $\kappa>0$ and any $\varepsilon>0$, the mixing time of $\UDbis(a,b,\kappa)$ at level $\varepsilon$ is finite.
\end{alem}
\begin{proof}
 In order to prove the existence of a unique limit distribution and argue that the mixing times are finite, we will make use of Foster-Lyapunov criteria discussed by Meyn and Tweedy in \cite{meyntweedie}. In particular, we will use the concept of super Lyapunov function, developed by Athreya, Kolba, and Mattingly in \cite{superLyapunov}, to which Appendix~\ref{sec:super} is devoted.
 
 Since the components $X^\kappa(\cdot,i)$ are distributed as independent continuous-time Markov chains, in order to prove robustness of $X^\kappa(\cdot)$ it is sufficient to prove it for all its components separately. The same holds for the description of the irreducible closed sets of $X^\kappa(\cdot)$, which are necessarily Cartesian products of closed and irreducible sets of the components $X^\kappa(\cdot,i)$. Finally, the mixing times of $X^\kappa(\cdot)$ at level $\varepsilon>0$ are finite for all $\varepsilon>0$, if and only if the mixing times $\tau^{\varepsilon,i}$ of $X^\kappa(\cdot,i)$ at level $\varepsilon>0$ are finite for all $\varepsilon>0$ and for all $i\in\{1,\dots,d\}$. 
 
 If $b(i)=0$, then the process $X^\kappa(\cdot,i)$ is the continuous-time Markov chain associated with the CRN with reaction diagram
 $$V_i \xrightarrow{\kappa^i_2} 0, \qquad 2V_i \xrightarrow{\kappa^i_3} 0.$$
 As such, $X^\kappa(\cdot,i)$ can only decrease. It follows that the set $\{0\}$ is closed and irreducible for $X^\kappa(\cdot,i)$. Moreover, it is the only closed and irreducible set, since the reaction $V_i\to 0$ can always take place as long as there is at least one molecule of $V_i$. 

 If $b(i)\neq0$, then the process $X^\kappa(\cdot,i)$ is the continuous-time Markov chain associated with the CRN with reaction diagram
 $$0 \xrightarrow{\kappa^i_1} V_i, \qquad (b(i)+1)V_i \xrightarrow{\kappa^i_2} a(i) V_i.$$
 Hence, the process $X^\kappa(\cdot,i)$ can always increase by 1, and decrease by $r_i$ if and only if at least $b(i)+1=r_i+a(i)$ molecules of $V_i$ are available. Hence, the set $\Theta(i)=\{v(i)\in\Int_{\geq0}\,|\, v(i)\geq a(i)\}$ is the only closed and irreducible set of $X^\kappa(\cdot,i)$.
 
 In both cases, we can conclude that $X^\kappa(\cdot,i)$ is robust by showing that a unique limit distribution exists. This, together with the fact that the mixing times $\tau^{\varepsilon, i}$ are finite, follows from Lemma~\ref{lem:superrangebis} and Theorem~\ref{thm:superstuff}.
\end{proof}

In Lemma~\ref{eherhieqjfeqjfeqjfleqlf} we proved that $X^{\kappa}(\cdot)$ is robust, which holds if and only if each process $X^{\kappa}(\cdot,i)$ is robust. We denote by $\pi^{\kappa}$ the unique limit distribution of $X^{\kappa}(\cdot)$, and by $\pi^{\kappa,i}$ the unique limit distribution of $X^{\kappa}(\cdot,i)$. Since the components of $X^{\kappa}(\cdot)$ are independent, it follows that
$$\pi^\kappa(v)=\prod_{i=1}^d \pi^{\kappa,i}(v(i))\quad\text{for all }v\in\Int_{\geq0}.$$
Hence, in order to study $\pi^\kappa$ it is sufficient to study the distributions $\pi^{\kappa,i}$, for $i\in\{1,\dots,d\}$, and this is what we will do.

\begin{alem}\label{pmd_at_0}
 Let $i\in\{1,\dots,d\}$ with $b(i)=0$. Then, $\pi^{\kappa,i}$ is the point mass distribution at 0.
\end{alem}
\begin{proof}
 The result follows from the fact that $X^{\kappa}(\cdot,i)$ is robust, and $\{0\}$ is its only closed irreducible set, as proven in Lemma~\ref{eherhieqjfeqjfeqjfleqlf}.
\end{proof}

\begin{alem}\label{lem1}
Let $i\in\{1,\dots,d\}$ with $b(i)\geq1$. Let $v_0(i) \ge b(i)+1$. 
Assume that
\begin{equation}\label{eq:constrwkewfwef}
\kappa^i_2r_i(b(i)+1)!>\kappa^i_1.
\end{equation}
Then
$$E^{\kappa,i}_{v_0(i)}[\sigma^{\kappa,i}] \le \frac{v_0(i)}{\kappa^i_2r_i(b(i)+1)!-\kappa^i_1}.$$
\end{alem}
\begin{proof}
Let $L(v(i)) = v(i)$ for all $v(i)\in\Int_{\geq0}$. By Dynkin's formula we have that for all $t>o$
\begin{align*}
E^{\kappa,i}_{v_0(i)}[L(X^\kappa(\min\{t, \sigma^{\kappa,i}\},i))] &= L(v_0(i)) + E^{\kappa,i}_{v_0(i)}\left[\int_0^{\min\{t, \sigma^{\kappa,i}\}} A^{\kappa,i}L(X^\kappa(s,i)) ds\right]\\
&\le v_0(i) + E^{\kappa,i}_{v_0(i)}\left[\int_0^{\min\{t, \sigma^{\kappa,i}\}} (\kappa^i_1 - \kappa^i_2 r_i(b(i)+1)!) ds\right]\\
&=  v_0(i) + (\kappa^i_1 - \kappa^i_2 r_i(b(i)+1)!) E^{\kappa,i}_{v_0(i)}[\min\{t, \sigma^{\kappa,i}\}].
\end{align*}
By \eqref{eq:constrwkewfwef} and since for all times $t>0$ we have $E^{\kappa,i}_{v_0(i)}[L(X^\kappa(\min\{t, \sigma^{\kappa,i}\},i))]\geq 0$, we may conclude
$$E^{\kappa,i}_{v_0(i)}[\min\{t,\sigma^{\kappa,i}\}] \le \frac{x_0(i)}{\kappa^i_2 r_i(b(i)+1)!-\kappa^i_1}.$$
We may use the monotone convergence theorem to conclude the proof.
\end{proof}

\begin{alem}\label{lem:pi_beta_converges}
Let $i\in\{1,\dots,d\}$ with $b(i)\geq1$, and assume \eqref{eq:constrwkewfwef} holds. Then,
\begin{align*}0\leq \frac1{r_i}-\pi^{\kappa,i}(b(i))&\leq \frac{1}{(r_i)^2}\cdot\frac{(\kappa^i_1)^2 }{\kappa^i_2(b(i)+1)!+\kappa^i_1}\left(\frac{b(i)+2}{\kappa^i_2r_i(b(i)+1)!-\kappa^i_1} + \frac{b(i)-a(i)}{\kappa^i_1}\right).
\end{align*}
\end{alem}
\begin{proof}
 By classical theory on continuous-time Markov chains it is known that
 $$\pi^{\kappa,i}(b(i)) = \frac{1/\kappa^i_1}{E^{\kappa,i}_{b(i)}[\sigma^{\kappa,i}_{b(i)}]},$$
 where $1/\kappa^i_1$ is the expectd holding time of $X^\kappa(\cdot,i)$ in the state $b(i)$. By conditioning on the first two steps, we have 
 \begin{align*}
	E^{\kappa,i}_{b(i)}[\sigma^{\kappa,i}_{b(i)}] &= \frac{1}{\kappa^i_1} + E^{\kappa,i}_{b(i)+1}[\sigma^{\kappa,i}_{b(i)}]\\
	&=\frac{1}{\kappa^i_1} + \frac{1}{\kappa^i_2(b(i)+1)!+\kappa^i_1} + \frac{\kappa^i_2(b(i)+1)!}{\kappa^i_2(b(i)+1)!+\kappa^i_1} E^{\kappa,i}_{a(i)}[\sigma^{\kappa,i}_{b(i)}]+\frac{\kappa^i_1}{\kappa^i_2(b(i)+1)!+\kappa^i_1} E^{\kappa,i}_{b(i)+2}[\sigma^{\kappa,i}_{b(i)}]\\
	&=\frac{1}{\kappa^i_1} + \frac{1}{\kappa^i_2(b(i)+1)!+\kappa^i_1} + \frac{\kappa^i_2(b(i)+1)!}{\kappa^i_2(b(i)+1)!+\kappa^i_1}\cdot \frac{b(i)-a(i)}{\kappa^i_1} + \frac{\kappa^i_1}{\kappa^i_2(b(i)+1)!+\kappa^i_1} E^{\kappa,i}_{b(i)+2}[\sigma^{\kappa,i}_{b(i)}]\\
	&=\frac{r_i}{\kappa^i_1} + \frac{\kappa^i_1}{\kappa^i_2(b(i)+1)!+\kappa^i_1} E^{\kappa,i}_{b(i)+2}[\sigma^{\kappa,i}_{b(i)}].
 \end{align*}
 Hence,
 \begin{align*}
     \frac{1}{r_i}-\pi^{\kappa,i}(b(i))&=\frac{1}{r_i}\left(1-\frac{r_i}{\kappa^i_1E^{\kappa,i}_{b(i)}[\sigma^{\kappa,i}_{b(i)}] }\right)=
 \frac{1}{r_i}\left(\frac{\frac{(\kappa^i_1)^2}{\kappa^i_2(b(i)+1)!+\kappa^i_1} E^{\kappa,i}_{b(i)+2}[\sigma^{\kappa,i}_{b(i)}]}{r_i+\frac{(\kappa^i_1)^2}{\kappa^i_2(b(i)+1)!+\kappa^i_1} E^{\kappa,i}_{b(i)+2}[\sigma^{\kappa,i}_{b(i)}]}\right)\\
 &=\frac{1}{(r_i)^2}\left(\frac{\frac{(\kappa^i_1)^2}{\kappa^i_2(b(i)+1)!+\kappa^i_1} E^{\kappa,i}_{b(i)+2}[\sigma^{\kappa,i}_{b(i)}]}{1+\frac{(\kappa^i_1)^2}{r_i\kappa^i_2(b(i)+1)!+r_i\kappa^i_1} E^{\kappa,i}_{b(i)+2}[\sigma^{\kappa,i}_{b(i)}]}\right)
 \leq \frac{1}{(r_i)^2}\cdot\frac{(\kappa^i_1)^2E^{\kappa,i}_{b(i)+2}[\sigma^{\kappa,i}_{b(i)}]}{\kappa^i_2(b(i)+1)!+\kappa^i_1}.
 \end{align*}
 From the first two equalities we have $\frac{1}{r_i}-\pi^{\kappa,i}(b(i))\geq0$, which is the first part of the lemma. Moreover, 
 $$
 E^{\kappa,i}_{b(i)+2}[\sigma^{\kappa,i}_{b(i)}]\le E^{\kappa,i}_{b(i)+2}[\sigma^{\kappa,i}] + \frac{b(i)-a(i)}{k^i_1},
 $$
 where the second term bounds from above the expected time to reach $b(i)$ from within the set $\{a(i),\dots,b(i)\}$. Hence, by Lemma \ref{lem1}, 
 $$\frac{1}{r_i}-\pi^{\kappa,i}(b(i))\leq \frac{1}{(r_i)^2}\cdot\frac{(\kappa^i_1)^2 }{\kappa^i_2(b(i)+1)!+\kappa^i_1}\left(\frac{b(i)+2}{\kappa^i_2r_i(b(i)+1)!-\kappa^i_1} + \frac{b(i)-a(i)}{\kappa^i_1}\right),$$
 which concludes the proof.
 \end{proof}
 
 \begin{alem}\label{lem:bound_on_external_pi}
  Let $i\in\{1,\dots,d\}$ with $b(i)\geq1$, and assume \eqref{eq:constrwkewfwef} holds. Let $w\in\Int_{\geq0}$ such that $w\in[2, r_i]$. Then,
 \[ 
  \pi^{\kappa,i}(b(i) + w) \le \left(\frac{\kappa^i_1}{\kappa^i_2}\right)^w C(i,w),
 \]
 where
 \begin{equation*}
     C(i,w)=\frac{\prod_{j=1}^{w}(j-1)!}{\prod_{j=1}^{w}(j+b(i))!}\leq \frac{1}{(b(i)+1)!(b(i)+2)!}.
 \end{equation*}
 \end{alem}
 \begin{proof}
  To simplify the notation, denote 
  $$\phi_i(v(i))=\frac{v(i)!}{(v(i)-b(i)-1)!}$$
  for all integers $v(i)$ that are greater than or equal to $b(i)+1$, so that the transition rate from $v(i)$ to $v(i)-r_i$ is given by $\kappa^i_2\phi_i(v(i)$. Then, by using again classical Markov chain theory we have
  \begin{equation}\label{fiheifhe}
      \pi^{\kappa,i}(b(i)+w) = \frac{1}{\kappa^i_2\phi_i(b(i)+w) + \kappa^i_1}\cdot\frac{1}{E^{\kappa,i}_{b(i)+w}[\sigma^{\kappa,i}_{b(i)+w}]},
  \end{equation}
  where the first factor is the expectation of the holding time of $X^{\kappa}(\cdot,i)$ in the state $b(i)+w$. By performing first step analysis on $E^{\kappa,i}_{b(i)+w}[\sigma^{\kappa,i}_{b(i)+w}]$, we have
  \begin{equation}\label{rvlirtnvilgidg}
  \begin{split}
  \begin{aligned}
   E^{\kappa,i}_{b(i)+w}[\sigma^{\kappa,i}_{b(i)+y}]&=\frac{\kappa^i_2\phi_i(b(i)+w)}{\kappa^i_2\phi_i(b(i)+w) + \kappa^i_1}E^{\kappa,i}_{a(i)+w-1}[\sigma^{\kappa,i}_{b(i)+w}]+\frac{\kappa^i_1}{\kappa^i_2\phi_i(b(i)+w) + \kappa^i_1}E^{\kappa,i}_{b(i)+w+1}[\sigma^{\kappa,i}_{b(i)+w}]\\
   &\geq \frac{\kappa^i_2\phi_i(b(i)+w)}{\kappa^i_2\phi_i(b(i)+w) + \kappa^i_1}E^{\kappa,i}_{a(i)+w-1}[\sigma^{\kappa,i}_{b(i)+w}].
  \end{aligned}
  \end{split}
  \end{equation}
  We will then study $E_{a(i)+w-1}[\sigma^{\kappa,i}_{b(i)+w}]$. Note that $a(i)+w-1\in[a(i), b(i)]$, because $2\leq w\leq r_i$ by assumption. 
  
  From each state $v(i)\in[a(i),b(i)]$, it takes at least one exponential time with rate $\kappa^i_1$ to reach the state $b(i)+1$. From $b(i)+1$, we reach the state $b(i)+w$ without hitting $[a(i), b(i)]$ with probability
  $$p^{\kappa, i}_w=\frac{(\kappa^i_1)^{w-1}}{(\kappa^i_2)^{w-1}\prod_{j=1}^{w-1} [\phi_i(j+b(i))+(\kappa^i_1/\kappa^i_2)]}.$$
  Hence, if we let $G^{\kappa, i}_w$ be a geometric random variable with parameter $p^{\kappa, i}_w$, we have
  \begin{equation}\label{wh2fl4hfifkb}
      E_{a(i)+w-1}[\sigma^{\kappa,i}_{b(i)+w}]\geq \frac{1}{\kappa^i_1}E[G^{\kappa, i}_w]=\frac{1}{\kappa^i_1 p^{\kappa, i}_w}=\frac{(\kappa^i_2)^{w-1}\prod_{j=1}^{w-1} [\phi_i(j+b(i))+(\kappa^i_1/\kappa^i_2)]}{(\kappa^i_1)^w}.
  \end{equation}
  Hence, by combining \eqref{fiheifhe}, \eqref{rvlirtnvilgidg}, and \eqref{wh2fl4hfifkb}, we have
  \begin{align*}
   \pi^{\kappa,i}(b(i)+w)&\leq \frac{1}{\kappa^i_2\phi_i(b(i)+w)E^{\kappa,i}_{b(i)+w+1}[\sigma^{\kappa,i}_{b(i)+w}]}\\
   &\leq\frac{(\kappa^i_1)^w}{(\kappa^i_2)^w \phi_i(b(i)+w)\prod_{j=1}^{w-1} [\phi_i(j+b(i))+(\kappa^i_1/\kappa^i_2)]}\\
   &\leq\frac{(\kappa^i_1)^w}{(\kappa^i_2)^w \prod_{j=1}^{w}\phi_i(j+b(i))},
  \end{align*}
  which concludes the proof.
 \end{proof}

 Finally, we are ready to prove the following result:
 \begin{athm}\label{thmskdjjgnejrrhgergejglkjblkrj}
  Let $i\in\{1,\dots,d\}$, and assume \eqref{eq:constrwkewfwef} holds. Let $q_i$ be the uniform distribution on the set $\{a(i),\dots, b(i)\}$. Then, for any integer $v(i)\in[a(i),b(i)]$ we have
  $$0\leq q_i(v(i))-\pi^{\kappa,i}(v(i))\leq \frac{\kappa^i_1}{\kappa^i_2}\mathbbm{1}_{\{b(i)\geq1\}}D(\kappa,i),$$
  where
  \begin{multline}\label{eq:D}
      D(\kappa,i)=\frac{1}{(b(i)+1)!}\left(\frac{1}{r_i}+\mathbbm{1}_{\{r_i\geq 2\}}\frac{1}{b(i)+2}\right)\\
      +\frac{\kappa^i_1}{\kappa^i_2}\left(\frac{1}{(r_i)^2}\cdot\frac{b(i)+2}{[(b(i)+1)!]^2-(\kappa^i_1/\kappa^i_2)^2}+\mathbbm{1}_{\{r_i\geq 3\}}\frac{1}{(b(i)+1)!(b(i)+2)!}\cdot \frac{1-(\kappa^i_1/\kappa^i_2)^{r_i-2}}{1-(\kappa^i_1/\kappa^i_2)}\right).
  \end{multline}
  Moreover,
  $$\|\pi^{\kappa,i}-q_i\|_\infty<\frac{\kappa^i_1}{\kappa^i_2}\mathbbm{1}_{\{b(i)\geq1\}}r_iD(\kappa,i).$$
 \end{athm}
  
 \begin{proof}
  By Lemma~\ref{pmd_at_0}, if $b(i)=0$ then $\pi^{\kappa,i}=q_i$, hence the statement holds. We now assume $b(i)\geq 1$.
  
  By Lemma~\ref{eherhieqjfeqjfeqjfleqlf}, $X^\kappa(\cdot,i)$ is non-explosive, hence the forward Kolmogorov equation holds true \cite{norris}. It follows that, provided that $r_i\geq2$ (which is equivalent to $b(i)\geq a(i)+1$), for any integer $v(i)\in[a(i)+1,b(i)]$
  $$\kappa^i_1\pi^{\kappa,i}(v(i)) = \kappa^i_1\pi^{\kappa,i}(v(i)-1)+\kappa^i_2\frac{(v(i)+r_i)!}{(v(i)-a(i))!}\pi^{\kappa, i}(v(i)+r_i).$$
  From $v(i)\in[a(i)+1,b(i)]$ it follows $v(i)+r_i\in [b(i)+2,b(i)+r_i]$. Hence, by Lemma \ref{lem:bound_on_external_pi} we have
  $$0\leq \pi^{\kappa,i}(v(i))-\pi^{\kappa,i}(v(i)-1) \leq \frac{1}{(b(i)+1)!(b(i)+2)!}\cdot\frac{(v(i)+r_i)!}{(v(i)-a(i))!} \left(\frac{\kappa^i_1}{\kappa^i_2}\right)^{v(i)-a(i)} .$$
  Hence, by applying the above inequality iteratively, we have that for any two integers $v_1(i)\leq v_2(i)\in[a(i), b(i)]$
  \begin{equation}\label{dnewbleqvlievnjkcsjvcec}
  \begin{split}
  \begin{aligned}
      0\leq \pi^{\kappa,i}(v_2(i))-&\pi^{\kappa,i}(v_1(i))\leq \pi^{\kappa,i}(b(i))-\pi^{\kappa,i}(a(i))\leq\frac{1}{(b(i)+1)!(b(i)+2)!}\sum_{j=a(i)+1}^{b(i)}\frac{(j+r_i)!}{(j-a(i))!} \left(\frac{\kappa^i_1}{\kappa^i_2}\right)^{j-a(i)}\\
      &= \mathbbm{1}_{\{r_i\geq 2\}}\frac{\kappa^i_1}{\kappa^i_2}\cdot\frac{1}{(b(i)+2)!}+\left(\frac{\kappa^i_1}{\kappa^i_2}\right)^2\cdot\frac{1}{(b(i)+1)!(b(i)+2)!}\sum_{j=0}^{r_i-3}\left(\frac{\kappa^i_1}{\kappa^i_2}\right)^{j}\\
      &= \mathbbm{1}_{\{r_i\geq 2\}}\frac{\kappa^i_1}{\kappa^i_2}\cdot\frac{1}{(b(i)+2)!}+\mathbbm{1}_{\{r_i\geq 3\}}\left(\frac{\kappa^i_1}{\kappa^i_2}\right)^2\frac{1}{(b(i)+1)!(b(i)+2)!}\cdot \frac{1-(\kappa^i_1/\kappa^i_2)^{r_i-2}}{1-(\kappa^i_1/\kappa^i_2)}.
  \end{aligned}
  \end{split}
  \end{equation}
By Lemma~\ref{eherhieqjfeqjfeqjfleqlf}, for any integer $v(i)<a(i)$ we have $\pi^{\kappa,i}(v(i))=q_i(v(i))=0$. 
By Lemma \ref{lem:pi_beta_converges} and \eqref{dnewbleqvlievnjkcsjvcec}, for any integer $v(i)\in[a(i),b(i)]$ 
\begin{equation}\label{eq:irhrekvlarevebver}
\begin{split}
  \begin{aligned}
      0\leq q_i(v(i))&-\pi^{\kappa,i}(v(i))= \frac{1}{r_i}-\pi^{\kappa,i}(b(i))+\pi^{\kappa,i}(b(i))-\pi^{\kappa,i}(v(i))\\
      &\leq \frac{1}{(r_i)^2}\cdot\frac{(\kappa^i_1)^2 }{\kappa^i_2(b(i)+1)!+\kappa^i_1}\left(\frac{b(i)+2}{\kappa^i_2r_i(b(i)+1)!-\kappa^i_1} + \frac{b(i)-a(i)}{\kappa^i_1}\right)\\
      &\quad+\mathbbm{1}_{\{r_i\geq 2\}}\frac{\kappa^i_1}{\kappa^i_2}\cdot\frac{1}{(b(i)+2)!}+\mathbbm{1}_{\{r_i\geq 3\}}\left(\frac{\kappa^i_1}{\kappa^i_2}\right)^2\frac{1}{(b(i)+1)!(b(i)+2)!}\cdot \frac{1-(\kappa^i_1/\kappa^i_2)^{r_i-2}}{1-(\kappa^i_1/\kappa^i_2)}\\
      &\leq \frac{\kappa^i_1}{\kappa^i_2}D(\kappa,i),
  \end{aligned}
  \end{split}
  \end{equation}
  where $D(\kappa,i)$ is as in \eqref{eq:D}.  Finally, from \eqref{eq:irhrekvlarevebver} it follows that
  \begin{align*}
  \sum_{v(i)=b(i)+1}^\infty \pi^{\kappa,i}(v(i))&=1-\sum_{v(i)=a(i)}^{b(i)} \pi^{\kappa,i}(v(i))\\
  &= \sum_{v(i)=a(i)}^{b(i)}\Big(q_i(v(i))-\pi^{\kappa,i}(v(i))\Big)\\
  &\leq \frac{\kappa^i_1}{\kappa^i_2}r_iD(\kappa,i).
  \end{align*}
  Hence, for any $v(i)\geq b(i)+1$
  $$|\pi^{\kappa,i}(v(i))-q_i(v(i))|=\pi^{\kappa,i}(v(i))\leq \frac{\kappa^i_1}{\kappa^i_2}r_i D(\kappa,i),$$
  which concludes the proof.
 \end{proof}
  
\subsection{Proof of Proposition~\ref{prop:uniform}}\label{dsflksdfsdlfkjsdfsdlkfjsldjjflsdjflksdjfl}

The CRN $\UD(a,b,\delta)$ is equal to $\UDbis(a,b,\kappa^\delta)$, with $\kappa^\delta$ as described in Construction~\ref{constr:range}. The quantity $D(\kappa^\delta,i)$ defined in \eqref{eq:D} becomes
$$D(\kappa^\delta,i)=\frac{1}{(b(i)+1)!}\left(\frac{1}{r_i}+\mathbbm{1}_{\{r_i\geq 2\}}\frac{1}{b(i)+2}\right)+\frac{\delta}{2d}g_i(\delta),$$
where $g_i(\delta)$ is a continuous function satisfying
$$\lim_{\delta\to 0}g_i(\delta)=\frac{1}{[(b(i)+1)!]^2}\left(\frac{1}{(r_i)^2}+\mathbbm{1}_{\{r_i\geq 3\}}\frac{1}{b(i)+2}\right).$$

Hence, by Theorem~\ref{thmskdjjgnejrrhgergejglkjblkrj} for any $v\in\Int_{\geq0}^d$ such that $a\leq v\leq b$ we have
\begin{align*}
    |\pi^{\delta}(v)-q(v)|&=q(v)-\pi^\delta(v)=\prod_{i=1}^d\left[\pi^{\kappa^\delta,i}(v(i))+\left(q_i(v(i))-\pi^{\kappa^\delta,i}(v(i))\right)\right]-\prod_{i=1}^d\pi^{\kappa^\delta,i}(v(i))\\ &\leq\prod_{i=1}^d\left[\frac{1}{r_i}+\frac{\delta}{2d}\mathbbm{1}_{\{b(i)\geq1\}}D(\kappa^\delta,i)\right]-\prod_{i=1}^d\frac{1}{r_i}\\
    &\leq \frac{\delta}{2d}\left(\frac{1}{\prod_{i=1}^d r_i}\sum_{i=1}^d \mathbbm{1}_{\{b(i)\geq1\}}\frac{1}{(b(i)+1)!}\left(1+\mathbbm{1}_{\{r_i\geq 2\}}\frac{r_i}{b(i)+2}\right)\right)+\delta^2G(\delta),
\end{align*}
where $G(\delta)$ is a continuous function satisfying $\lim_{\delta\to 0}G(\delta)<\infty.$
If $v\in\Int_{\geq0}^d$ does not satisfy $v\geq a$, then by Lemma~\ref{eherhieqjfeqjfeqjfleqlf} we have $\pi^\delta(v)=0=q(v)$. Finally, for all $v\in\Int_{\geq0}^d$ with $v\geq a$ and $v\nleq b$, we have
\begin{align*}
    |\pi^{\delta}(v)-q(v)|&=\pi^{\delta}(v)\leq \sum_{v'\geq a, v'\nleq b}\pi^{\delta}(v')=1-\sum_{a\leq v'\leq b}\pi^{\delta}(v')=\sum_{a\leq v'\leq b}\left(q(v')-\pi^{\delta}(v')\right)\\
    &\leq\frac{\delta}{2d} \sum_{i=1}^d \mathbbm{1}_{\{b(i)\geq1\}}\frac{1}{(b(i)+1)!}\left(1+\mathbbm{1}_{\{r_i\geq 2\}}\frac{r_i}{b(i)+2}\right)+\delta^2 G(\delta)\prod_{i=1}^d r_i.
\end{align*}
The proof is concluded by noting that $r_i<b(i)+2$.\hfill\qed

\subsection{Analysis of Construction~\ref{constr:poisson}}

Let $X(\cdot)$ be the continuous-time Markov chain associated with $\PPN(c)$. Note that the different components $X(\cdot, i)$ are independent continuous-time Markov chains, each one associated with the subnetwork of $\PPN(c)$ governing the changes of the species $V_i$. We state an prove the following results concerning Construction~\ref{constr:poisson}.

\begin{alem}\label{lem:superppn}
  Let $d\geq1$ and let $c\in\Rel_{>0}^d$. Consider the function $L(\cdot)$, defined as $L(v(i))=v(i)$ for all $v(i)\in \Int_{\geq0}$. Then, $\lim_{v(i)\to\infty}L(v(i))=\infty$. Moreover, for any $i\in\{1,\dots,d\}$, there exists $\alpha_i\in\Rel_{>0}$ and a compact set $K\subset \Int_{\geq 0}$ such that
  $$A^i L(v(i))\leq -\alpha_i L^2(v(i))\quad\text{for all }v(i)\notin K,$$
  where $A^i$ is the generator of $X(\cdot, i)$. 
\end{alem}
\begin{proof}
 It is clear that $\lim_{v(i)\to\infty}L(v(i))=\infty$. Moreover, 
 $$A^i L(v(i))=c(i)^2+c(i)v(i)-2v(i)\Big(v(i)-1\Big).$$
 Hence, if $v(i)$ is large enough we have
 $$A^i L(v(i))\leq -v(i)^2=-L^2(v(i)),$$
 which concludes the proof.
\end{proof}

\begin{aprop}\label{prop:ppn}
 The CRN $\PPN(c)$ is robust, with limit distribution $\pi$ given by
 \begin{equation}\label{aaaaaaaaaaaaaaaaaaaaaaaaaaaaaaaaaaaaaaa}
     \pi(v)=\prod_{i=1}^d e^{-c(i)}\frac{(c(i))^{v(i)}}{v(i)!}\quad\text{for all }v\in\Int_{\geq0}^d.
 \end{equation}
 Moreover, the mixing times of $\PPN(c)$ are finite at any level $\varepsilon>0$.
\end{aprop}
\begin{proof}
Since the components $X(\cdot, i)$ are independent, the state space of $X(\cdot)$ is irreducible if and only if the same holds for each $X(\cdot, i)$. This is the case, since the molecules of a species $V_i$ can always increase by 1, and to decrease by 2 whenever at least 2 molecules are available. 

The CRN $\PPN(c)$ is \emph{complex balanced} with complex balanced equilibrium $c$, in the sense of \cite{horn1972general}. Indeed, for any complex $y\in\Cmp$ it holds that
$$\sum_{\substack{y'\in\Cmp\\ y'\to y\in \Rel}} \kappa_{y'\to y}c^{y'}=\sum_{\substack{y'\in\Cmp\\ y\to y'\in \Rel}} \kappa_{y\to y'}c^{y}.$$
Hence, due to \cite{ACKK:explosion} the associated continuous-time Markov chain $X(\cdot)$ is non-explosive, and due to \cite{ACK:product} and to the fact that the state space is irreducible, the limit distribution satisfies \eqref{aaaaaaaaaaaaaaaaaaaaaaaaaaaaaaaaaaaaaaa}.

Finally, the mixing times of $X(\cdot)$ are finite at any level, if and only if the mixing times of any component $X(\cdot, i)$ are finite at any level. 
The latter is implied by Lemma~\ref{lem:superppn} and Theorem~\ref{thm:superstuff}, hence the proof is concluded.
\end{proof}

\subsection{Proof of Theorem~\ref{thm:final}}\label{sec:proof_thm_final}

Let $X^\delta(\cdot)$ be the continuous-time Markov chain associated with $\Mix(\Fam,\zeta,\delta)$. By Lemma~\ref{lem:pmd}, Proposition~\ref{prop:uniform}, and Proposition~\ref{prop:ppn}, all the networks $\Net^\delta_i$ are robust, for all $\delta>0$. Let $\pi_i^\delta$ denote the limit distribution of $\Net^\delta_i$. By Lemma~\ref{lem:pmd}, Proposition~\ref{prop:uniform}, and Proposition~\ref{prop:ppn} we have
\begin{align}\label{i1}
    \|\pi_i^\delta-\pi_i\|_\infty&\leq \delta\quad\text{for all }i\in I_1;\\
    \label{i2}
    \|\pi_i^\delta-\pi_i\|_\infty&\leq \delta+o(\delta)\quad\text{for all }i\in I_2;\\
    \label{i3}
    \|\pi_i^\delta-\pi_i\|_\infty&=0 \quad\text{for all }i\in I_3,
\end{align}
where $o(\delta)$ is a function with the property
$$\lim_{\delta\to 0}\frac{o(\delta)}{\delta}=0.$$
Furthermore, Constructions~\ref{constr:pmd} and \ref{constr:range} are special cases of Construction~\ref{constr:rangebis}, hence by Lemma~\ref{eherhieqjfeqjfeqjfleqlf} and Proposition~\ref{prop:ppn} it follows that the mixing times of any network $\Net^\delta_i$ are finite at any level $\varepsilon>0$, for any $\delta>0$.

If we can show that $\Mix(\Fam,\zeta,\delta)$ is non-explosive for any $\delta>0$, we can conclude the proof by Theorem~\ref{thm:general} and by triangular inequality, using \eqref{i1}, \eqref{i2}, and \eqref{i3}.

Let $X^\delta_\Hsp(\cdot)$ be the projection of $X^\delta(\cdot)$ onto the species $\{H_1, \dots, H_m\}$. Note that $X^\delta_\Hsp(\cdot)$ is a continuous-time Markov chain, associated with the CRN
\[
0 \xrightleftharpoons[\delta]{\delta^2 \zeta(i)} H_i, \quad \text{ for } i\in\{1,\ldots,m\}.
\]
The above CRN is detailed balanced, with detailed balanced equilibrium $\delta \zeta\in\Rel^m_{>0}$. Hence, due to \cite{ACKK:explosion} $X^\delta_\Hsp(\cdot)$ is non-explosive. If $X^\delta(\cdot)$ were explosive, then infinitely many transitions of $X^\delta(\cdot)$ would occur while $X^\delta_\Hsp(\cdot)$ is fixed at a state $h$. Note that given $X^\delta_\Hsp(\cdot)\equiv h$, the components of $X^\delta(\cdot)$ relative to the species $\{V_1, \dots, V_d\}$ are independent continuous-time Markov chains. Let $X^{\delta,h}(\cdot, j)$ denote the component of $X^\delta(\cdot)$ relative to species $V_j$, given that $X^\delta_\Hsp(\cdot)\equiv h$. If $X^\delta(\cdot)$ were explosive, then a process $X^{\delta,h}(\cdot, j)$ would be explosive, for some $h\in\Int_{\geq0}^m$ and some $j\in\{1,\dots,d\}$. We will conclude the proof by showing that this is not possible.

Let $A^{\delta,i}$ be the generator of $\Net_i^\delta$. Then, the generator of $X^{\delta,h}(\cdot, j)$ is given by
$$A^{\delta,h}=\sum_{i=1}^m h(i)A^{\delta,i}.$$
From Lemma~\ref{lem:superrangebis} and Lemma~\ref{lem:superppn}, it follows that if $L(v(j))=v(j)$ for all $v(j)\in\Int_{\geq0}$, then for each $i\in\{1,\dots,m\}$ and for each $\delta\in\Rel_{>0}$, there exist $\alpha^{\delta,i}\in\Rel_{>0}$ and a compact set $K^{\delta,i}\subset \Int_{\geq0}$ such that
$$A^{\delta,i} L(v(j))\leq -\alpha^{\delta,i} L^2(v(j))\quad\text{for all }v(j)\notin K^{\delta,i}, i\in\{1,\dots,m\}.$$
Hence, if
$$K^\delta =\bigcup_{i=1}^m K^{\delta,i}\quad\text{ and }\quad\alpha^\delta=\min_{1\leq i\leq m} h(i)\alpha^{\delta,i},$$
then
$$A^{\delta,h} L(v(j))=\sum_{i=1}^m h(i)A^{\delta,i} L(v(j))\leq -\alpha^\delta L^2(v(j))\quad\text{for all }v(j)\notin K^\delta.$$
The proof is then concluded by Theorem~\ref{thm:superstuff}.\hfill\qed


\subsection{Proof of Theorem~\ref{thm:mixing}}\label{sec:proof_thm_mixing}

If $q$ has finite support $\{x_1,\dots,x_m\}$, then we have
$$q=\sum_{i=1}^m q(x_i)\delta_{x_i}.$$
Note that $\PMCN(q,\delta)=\Mix\{\Fam, \zeta, \delta\}$, if we let $\Net_i^\delta=\PMD(x_i,\delta)$ and $\zeta(i)=q(x_i)$ for all $i\in\{1,\dots,m\}$. Hence, the proof is concluded by Theorem~\ref{thm:final}.

\section{Bounds for the mixing times of $\PMD(x,\delta)$}\label{sec:bounds}

Here we give some useful bounds on the mixing times of a generalization of one-dimensional $\PMD(x,\delta)$, where the choice of rate constants is not constrained.

\begin{aprop}
Consider the CRN
 \begin{gather*}
 V\xrightarrow{\kappa_1}0\\
 2V\xrightarrow{\kappa_2}0
\end{gather*}
Then, the CRN is robust with unique limit distribution $\pi$ being the point mass distribution centered at 0. Moreover, for any choice of rate constants $\kappa_1, \kappa_2$, and for all $\varepsilon>0$,
$$\tau^\varepsilon\leq \frac{1}{\varepsilon}\sum_{v=1}^{\infty}\frac{1}{\kappa_1v+\kappa_2v(v-1)}<\infty.$$
\end{aprop}
\begin{proof}
 Robustness of the CRN and the fact that the limit distribution is the point mass distribution at 0 follows from Lemma~\ref{pmd_at_0}. Let $X(\cdot)$ denote the continuous-time Markov chain associated with the CRN. Let 
 \begin{itemize}
  \item $\sigma=\inf\{t\geq0\,:\,X(t)=0\}$;
  \item $Y(\cdot)$ be the embedded discrete time Markov chain of $X(\cdot)$: $Y(0)=X(0)$ and for each $n\geq 1$, $Y(n)$ is defined as the value of the process $X(\cdot)$ after $n$ jumps;
  \item $M=\inf\{n\geq0:Y(n)=0\}$.
 \end{itemize}
 Let $X(0)=v_0$. Note that since the process $X(\cdot)$ decreases at least by one unit at each jump, necessarily $M\leq v_0$. Denote by $(\mathcal{E}_v)_{v=0}^\infty$ a sequence of independent exponential random variables with rates $\kappa_1v+\kappa_2v(v-1)$. Then,
 \begin{equation*}
  \sigma=\sum_{n=0}^{M-1}\mathcal{E}_{Y(n)}\leq\sum_{v=1}^{v_0}\mathcal{E}_{v}
 \end{equation*}
 Then, by Markov inequality,
 $$P(0,t|v_0)=P(\sigma\leq t|v_0)\geq 1-\frac{1}{t}\sum_{v=1}^{v_0}\frac{1}{\kappa_1v+\kappa_2v(v-1)}\geq 1-\frac{1}{t}\sum_{v=1}^{\infty}\frac{1}{\kappa_1v+\kappa_2v(v-1)}.$$
 The proof is then concluded by noting that in this case
 \begin{align*}
  \|P(\cdot,t|v_0)-\pi\|_\infty&\leq \max\left\{|P(0,t|v_0)-\pi(0)|,\sum_{v=1}^\infty|P(v,t|v_0)-\pi(v)|\right\}\\
  &= \max\left\{1-P(0,t|v_0),\sum_{v=1}^\infty P(v,t|v_0)\right\}\\
  &=1-P(0,t|v_0).
 \end{align*}
\end{proof}

\begin{aprop}
 Consider an integer $x\geq 1$, and consider the CRN
 \begin{gather*}
 0\xrightarrow{\kappa_1}V\\
 (x+1)V\xrightarrow{\kappa_2}xV
\end{gather*}
 Then, the CRN is robust. Let $\pi$ be its unique limit distribution. Moreover, assume that
 \begin{align}\label{dfvlkjdfnvkjdfnk}
     \varepsilon&>\max\left\{\frac{\kappa_1^2(x+2)}{\kappa^2_2[(x+1)!]^2-\kappa_1^2}, 1-e^{-\frac{\kappa_1}{x!\cdot x}}\right\}\quad&\text{if }&\kappa^2_2[(x+1)!]^2>\kappa_1^2\\
     \label{jhihkljwedwsfcwsdc}
     \varepsilon&>\max\left\{\frac{2 \kappa_1}{\kappa_2}, 1-e^{-\frac{\kappa_1}{x!\cdot x}}\right\}\quad&\text{if }&\kappa^2_2[(x+1)!]^2\leq\kappa_1^2.
 \end{align}
 Then
 $$\tau^{2\varepsilon}\leq \max\left\{\frac{e^{-\frac{\kappa_1}{x!\cdot x}}}{\kappa_2(e^{-\frac{\kappa_1}{x!\cdot x}}-1+\varepsilon)x!\cdot x}, \frac{x}{\kappa_1\varepsilon}\right\}.$$
\end{aprop}
\begin{proof}
 First, the CRN is robust due to Lemma~\ref{eherhieqjfeqjfeqjfleqlf}. Note that by Lemma~\ref{lem:pi_beta_converges} and Lemma~\ref{lem:pmd}, it follows from \eqref{dfvlkjdfnvkjdfnk} and \eqref{jhihkljwedwsfcwsdc} that
 \begin{equation}\label{wejeghelrivgdalfkjviaebvejncvj}
     \varepsilon>\max\left\{1-\pi(x), 1-e^{-\frac{\kappa_1}{x!\cdot x}}\right\}.
 \end{equation}
 Let $X(\cdot)$ be the process associated with the CRN, and let
 $$\sigma=\inf\{t\geq0\,:\, X(t)=x\}.$$
 Since $\{X(t)=x\}\subseteq \{\sigma\leq t\}$, we have that for any $v_0\in\Int_{\geq0}$
 $$P(x,t|v_0)=P(X(t)=x, \sigma\leq t | X(0)=v_0)= P(X(t)=x|\sigma\leq t,X(0)=v_0)P(\sigma\leq t|X(0)=v_0).$$
 By monotonicity of birth and death processes \cite[Section 2]{monotonicity} and by strong Markov property we have
 $$P(X(t)=n|\sigma<t,X(0)=v_0)\geq \pi(x).$$
 Hence,
 $$P(x,t|v_0)\geq \pi(x)P(\sigma\leq t|X(0)=v_0).$$ 
   There are three cases:
 \begin{enumerate}
  \item If $v_0<x$, and if we denote by $\Gamma(k, \theta)$ the sum of $k$ independent exponential random variables with mean $\theta$, then by Markov inequality
 \begin{align*}
  P(\sigma\leq t|X(0)=v_0)&=P\left(\Gamma\left(x-v_0, \frac{1}{\kappa_1}\right)\leq t\right)\\
  &\geq P\left(\Gamma\left(x, \frac{1}{\kappa_1}\right)\leq t\right)\\
  &\geq 1-\frac{x}{\kappa_1t}\\
 \end{align*}
  \item If $v_0=x$, then $P(\sigma<t|X(0)=v_0)=1$ for all $t\geq0$.
  \item If $v_0>x$, then 
  $$P(\sigma\leq t|X(0)=v_0)\geq P(\sigma\leq t|A_{v_0}, X(0)=v_0)P(A_{v_0}|X(0)=v_0),$$
  where
  $$A_{v_0}=\{\text{the first }v_0-x\text{ jumps of $X(\cdot)$ point downwards}\}.$$
  For notational convenience, let
  $$\phi^x(v)=\frac{v!}{(v-x-1)!}\quad\text{for }v\geq x+1,$$
  and denote by $(\mathcal{E}_v)_{v=x+1}^\infty$ a sequence of independent exponential random variables with rates $\kappa_2 \phi^x(v)$. Then, by Markov inequality,
  \begin{align*}
  P(\sigma\leq t|A_{v_0}, X(0)=v_0)&=P\left(\sum_{v=x+1}^{v_0}\mathcal{E}_v\leq t\right)\\
  &\geq 1-\frac{1}{\kappa_2 t}\sum_{v=x+1}^{v_0}\frac{1}{\phi^x(v)}\\
  &\geq 1-\frac{1}{\kappa_2 t}\sum_{v=x+1}^{\infty}\frac{1}{\phi^x(v)}.
  \end{align*}
  In order to express better the last series, note that
  \begin{align*}
  \sum_{v=x}^{\infty}\frac{(v-x)!}{v!}&=\frac{1}{x!}+\sum_{v=x+1}^{\infty}(v-x)\frac{(v-x-1)!}{v!}\\
  &=\frac{1}{x!}+\sum_{v=x}^{\infty}\frac{(v-x)!}{v!}-x\sum_{v=x+1}^{\infty}\frac{(v-x-1)!}{v!}\\
  \end{align*}
  It follows that
  $$\sum_{v=x+1}^{\infty}\frac{1}{\phi^x(v)}=\sum_{v=x+1}^{\infty}\frac{(v-x-1)!}{v!}=\frac{1}{x!\cdot x}$$
  and
  $$P(\sigma\leq t|A_{v_0}, X(0)=v_0)\geq 1-\frac{1}{x!\cdot x\kappa_2 t}.$$
  Moreover, 
  \begin{align*}
  P(A_{v_0}|X(0)=v_0)&=\prod_{v=x+1}^{v_0}\frac{\phi^x(v)}{\kappa_1+\phi^x(v)}\\
  &=e^{-\sum_{v=x+1}^{v_0}\log\left(1+\frac{\kappa_1}{\phi^x(v)}\right)}\\
  &\geq e^{-\sum_{v=x+1}^{\infty}\frac{\kappa_1}{\phi^x(v)}}\\
  &= e^{-\frac{\kappa_1}{x!\cdot x}},\\
  \end{align*}
  which implies
  $$P(\sigma\leq t|X(0)=v_0)\geq e^{-\frac{\kappa_1}{x!\cdot x}}\left(1-\frac{1}{x!\cdot x\kappa_2 t}\right)$$
 \end{enumerate}
Hence, independently on the initial condition $v_0$  and 
provided that
$$\varepsilon>1-e^{-\frac{\kappa_1}{x!\cdot x}},$$
if
$$t\geq \max\left\{\frac{e^{-\frac{\kappa_1}{x!\cdot x}}}{\kappa_2(e^{-\frac{\kappa_1}{x!\cdot x}}-1+\varepsilon)x!\cdot x}, \frac{x}{\kappa_1\varepsilon}\right\}$$
then
$$ P(x,t|v_0)\geq (1-\varepsilon)\pi(x).$$
Hence, since $\varepsilon>1-\pi(x)$ by \eqref{wejeghelrivgdalfkjviaebvejncvj}, it follows that
$$|P(x,t|v_0)-\pi(x)|\leq \varepsilon.$$
Moreover, for any $v\neq x$ we have
\begin{align*}
 |P(v,t|v_0)-\pi(v)|&\leq \max\{P(v,t|v_0),\pi(v)\}\\
 &\leq \max\{1-P(x,t|v_0),1-\pi(x)\}\\
 &\leq 1-\pi(x) +\varepsilon\pi(x)<2\varepsilon,
\end{align*}
which concludes the proof.
\end{proof}

\section{Super Lyapunov functions}\label{sec:super}

The theory we develop here was mostly developed in \cite{superLyapunov}, for a specific family of stochastic differential equations. The concept and the terminology of ``super Lyapunov function'' themselves were also introduced in \cite{superLyapunov}. We are interested in an adaptation of \cite[Lemma~6.1]{superLyapunov}, which we state here as Theorem~\ref{thm:super} and whose proof we repeat for completeness.

\begin{adefn}\label{def:super}
 Let $X(\cdot)$ be a continuous-time Markov chain on $\Int_{\geq0}^d$, with generator $A$. We say that a function $L\colon \Int_{\geq0}^d\to \Rel_{\geq0}$ is a \emph{super Lyapunov function} if the following holds true:
 \begin{itemize}
     \item $\lim_{x\to\infty} L(x)=\infty$
     \item there exists a compact set $K$ and real numbers $\alpha>0$ and $\gamma>1$ such that 
     \begin{equation}\label{eq:supercmp}
         AL(x)\leq -\alpha L^\gamma (x)\quad\text{for all }x\notin K.
     \end{equation}
 \end{itemize}
\end{adefn}
\begin{armk}\label{rem:SLF}
If \eqref{eq:supercmp} holds for $\gamma=1$, then the function $L$ is a standard Lyapunov function. While the existence of such a function implies that the process $X(\cdot)$ is non explosive and that a limit distribution exists for any initial condition \cite{meyntweedie}, in general it does not imply that the mixing times are finite.\hfill $\triangle$
\end{armk}
\begin{armk}
Equation \eqref{eq:supercmp} is equivalent to the existence of real numbers $\alpha, \beta\geq0$ and $\gamma>1$ such that
     \begin{equation}\label{eq:superbeta}
         AL(x)\leq -\alpha L^\gamma (x)+\beta\quad\text{for all }x\in \Int_{\geq0}^d.
     \end{equation}
Indeed, it is sufficient to consider $b=\max_{x\in K} |AL(x)|$.\hfill$\triangle$
\end{armk}
\begin{athm}\label{thm:super}
 Let $X(\cdot)$ be a continuous-time Markov chain on $\Int_{\geq0}^d$, and let $L(\cdot)$ be a super Lyapunov function. Then, 
 $$E[L(X(t))|X(0)=x_0]\leq \max\left\{\left(\frac{2\beta}{\alpha}\right)^{\frac{1}{\gamma}}, \left(\frac{2}{\alpha(\gamma-1)t}\right)^{\frac{1}{\gamma-1}}\right\}\quad\text{for all }t\geq0, x_0\in\Int_{\geq0}^d,$$
 where $\alpha, \beta, \gamma$ are as in \eqref{eq:superbeta}.
\end{athm}
\begin{proof}
 For simplicity, denote
 $$z_{x_0}(t)=E[L(X(t))|X(0)=x_0].$$
 Moreover, for any real number $M>0$ let $\rho^M$ be the stopping time
 $$\rho^M=\inf\{t\geq0\,:\,L(V(t))\geq M\}.$$
 By Dynkin's formula, \eqref{eq:superbeta} and Jensen's inequality, for any $M>L(x_0)$ we have
 \begin{align*}
     z_{x_0}(\max\{t,\rho^M\})&=L(x_0)+E\left[\left.\int_0^{\max\{t,\rho^M\}} AL(X(s))ds\right| X(0)=x_0\right]\\
     &\leq L(x_0)+\beta t-\alpha\int_0^t \left(z_{x_0}(s)\right)^\gamma ds.
 \end{align*}
 By taking the limit for $M$ going to infinity we have
 $$z_{x_0}(t)\leq L(x_0)+\beta t-\alpha\int_0^t \left(z_{x_0}(s)\right)^\gamma ds.$$
 Since the latter holds for any initial condition $x_0$, we must have
 \begin{equation}\label{wfhwiefh}
     \frac{d}{dt}z_{x_0}(t)\leq \beta-\alpha\left(z_{x_0}(t)\right)^\gamma\leq -\frac{\alpha}{2}\left(z_{x_0}(t)\right)^\gamma\quad\text{for }z_{x_0}(t)\geq\left(\frac{2\beta}{\alpha}\right)^{\frac{1}{\gamma}}.
 \end{equation}
 Since the latter is strictly negative, it follows that as soon as we have $z_{x_0}(t)\leq\left(\frac{2\beta}{\alpha}\right)^{\frac{1}{\gamma}}$ for some $t\geq0$, then the same holds for all times afterwards. In particular, if this holds for $t=0$, then the proof is complete. We can therefore assume that $z_{x_0}(0)\geq\left(\frac{2\beta}{\alpha}\right)^{\frac{1}{\gamma}}$. By \eqref{wfhwiefh}, it follows that as long as $z_{x_0}(t)\geq\left(\frac{2\beta}{\alpha}\right)^{\frac{1}{\gamma}}$, then $z_{x_0}(t)\leq u_{x_0}(t)$ where $u_{x_0}(t)$ is the solution to
 \[
    \frac{d}{dt}u_{x_0}(t)=-\frac{\alpha}{2}\left(u_{x_0}(t)\right)^\gamma,\quad u_{x_0}(0)=z_{x_0}(0)\geq\left(\frac{2\beta}{\alpha}\right)^{\frac{1}{\gamma}}.
\]
 The latter can be explicitly calculated, and we have
 $$u_{x_0}(t)=\left(\frac{\alpha(\gamma-1)t}{2}+u_{x_0}(0)^{1-\gamma}\right)^{\frac{1}{1-\gamma}}\leq \left(\frac{2}{\alpha(\gamma-1)t}\right)^{\frac{1}{\gamma-1}},$$
 which concludes the proof.
\end{proof}
\begin{athm}\label{thm:superstuff}
 Let $X(\cdot)$ be a continuous-time Markov chain on $\Int_{\geq0}^d$, and let $L(\cdot)$ be a super Lyapunov function. Then, $X(\cdot)$ is non-explosive. Moreover, if there exists a unique closed irreducible set, then $X(\cdot)$ has a unique limit distribution and the mixing times at any level are finite.
\end{athm}
\begin{proof}
 $X(\cdot)$ is non-explosive and admits a unique limit distribution $\pi$ (when a unique closed irreducible set exists) due Foster-Lyapunov theory \cite{meyntweedie} (see Remark~\ref{rem:SLF}). Hence, we only need to prove that the mixing times are finite. 
 Let $\tau^\varepsilon$ be the mixing time at level $\varepsilon$, and for any $x\in\Int_{\geq0}^d$ let
 $$\tau_x^{\frac{\varepsilon}{2}}=\inf\left\{t\geq0\,:\, \|P(\cdot,s|x)-\pi\|_\infty<\frac{\varepsilon}{2}\quad\text{for all }s\geq t\right\}.$$
 Note that since $\pi$ is a limit distribution, for any $x\in\Int^d_{\geq0}$ we have $\tau_x^{\frac{\varepsilon}{2}}<\infty$. Let
 $$M^\varepsilon=\frac{4}{\varepsilon}\left(\frac{2\beta}{\alpha}\right)^{\frac{1}{\gamma}}.$$
 The set $\Xi^\varepsilon=\{x\in\Int_{\geq0}^d\,:\, L(x)\leq M^\varepsilon\}$ is finite, because $\lim_{x\to\infty}L(x)=\infty$. Hence,
 $$R^\varepsilon\doteq\max_{x\in\Xi^\varepsilon}\tau_x^{\frac{\varepsilon}{2}}<\infty.$$
 Finally, let
 $$T=\frac{2}{\alpha(\gamma-1)}\left(\frac{\alpha}{2\beta}\right)^{\frac{\gamma-1}{\gamma}}$$
 By Theorem~\ref{thm:super}, we have that for any $x_0\in\Int^d_{\geq0}$
 $$E[L(X(T))|X(0)=x_0]\leq \left(\frac{2\beta}{\alpha}\right)^{\frac{1}{\gamma}},$$
 and by Markov inequality we have that for any $x_0\in\Int^d_{\geq0}$
 $$P(X(T)\notin\Xi^\varepsilon|x_0)=P(L(X(T))>M^\varepsilon|x_0)\leq \frac{E[L(X(T))|X(0)=x_0]}{M^\varepsilon}\leq\frac{\varepsilon}{4}.$$
 Hence, by Markov property and by definition of $R^\varepsilon$, for any $s\geq T+R^\varepsilon$ and any $x_0\in\Int_{\geq0}^d$ we have
 $$\|P(\cdot,s|x_0)-\pi\|_\infty \leq \frac{\varepsilon}{4}\cdot 2+\max_{x\in\Xi^\varepsilon}\|P(\cdot,s-T|x)-\pi\|_\infty\leq \varepsilon.$$
 It follows that $\tau^\varepsilon\leq T+R^\varepsilon$, which concludes the proof.
\end{proof}
\end{document}